\documentclass{amsart}

\usepackage{amssymb}
\usepackage{amsthm}
\usepackage{amsfonts}
\usepackage{amsmath}
\usepackage{pmboxdraw}
\usepackage{verbatim} 
\usepackage{graphicx}
\usepackage{color}
\usepackage[colorlinks=true, citecolor=blue, filecolor=black, linkcolor=black, urlcolor=black]{hyperref}
\usepackage{cite}
\usepackage[normalem]{ulem}
\usepackage{subcaption}
\usepackage{bbm}
\usepackage{bm}
\usepackage{mathtools}
\usepackage{mathrsfs}
\usepackage{todonotes}
\usepackage{kantlipsum}
\allowdisplaybreaks

\usepackage{xparse}

\ExplSyntaxOn
\NewDocumentCommand{\MeijerG}{smmmm}
 {
  \IfBooleanTF{#1}
   {
    \vic_meijerg:nnnnnn { #2 } { #3 } { #4 } { #5 } { small } { }
   }
   {
    \vic_meijerg:nnnnnn { #2 } { #3 } { #4 } { #5 } { } { \; }
   }
 }

\seq_new:N \l__vic_meijerg_args_in_seq
\seq_new:N \l__vic_meijerg_args_out_seq

\cs_new_protected:Nn \vic_meijerg:nnnnnn
 {
  \seq_set_split:Nnn \l__vic_meijerg_args_in_seq { | } { #3 }
  \seq_clear:N \l__vic_meijerg_args_out_seq  
  \seq_map_inline:Nn \l__vic_meijerg_args_in_seq
   {
    \seq_put_right:Nn \l__vic_meijerg_args_out_seq
     {
      \begin{#5matrix} ##1 \end{#5matrix}
     }
   }
  G\sp{#1}\sb{#2}
  \left(
  \seq_use:Nn \l__vic_meijerg_args_out_seq { #6\middle|#6 }
  #6\middle|#6
  #4
  \right)
 }
\ExplSyntaxOff

\newcommand{\RN}[1]{%
	\textup{\uppercase\expandafter{\romannumeral#1}}%
}

\def\cF{\mathcal{F}}
\def\cL{\mathcal{L}}

\def\cK{\mathcal{K}}

\def\C{\mathbb{C}}

\def\E{\mathbf{E}}
\def\N{\mathbb{N}}
\def\P{\mathbf{P}}
\def\R{\mathbb{R}}

\def\zbar{\overline{z}}
\def\wbar{\overline{w}}

\newcommand{\erfc}{\operatorname{erfc}}

\newcommand{\re}{\operatorname{Re}}
\newcommand{\im}{\operatorname{Im}}

\theoremstyle{plain}
\newtheorem*{thm*}{Theorem}
\newtheorem{thm}{Theorem}[section]
\newtheorem{lem}[thm]{Lemma}
\newtheorem{cor}[thm]{Corollary}
\newtheorem{prop}[thm]{Proposition}
\newtheorem*{prop*}{Proposition}
\newtheorem*{lem*}{Lemma}
\newtheorem{rem}[thm]{Remark}

\theoremstyle{definition}
\newtheorem*{eg*}{Example}
\newtheorem*{egs*}{Examples}
\newtheorem*{Q*}{Question}

\theoremstyle{remark}
\newtheorem*{rmk*}{Remark}
\newtheorem*{rmks*}{Remarks}

\numberwithin{equation}{section}

\begin{document}
\title[Determinantal structure of the overlap of ISUE]{Determinantal structure of the overlaps for \\ 
induced spherical unitary ensemble
}

\author{Kohei Noda}
\address{Institute of Mathematics for Industry, Kyushu University, West Zone 1, 744 Motooka, Nishi-ku, Fukuoka 819-0395, Japan}
\email{k-noda@imi.kyushu-u.ac.jp}



\date{\today}

\begin{abstract}
In this note, we study the determinantal structure of the $k$-th conditional expectation of the overlap for induced spherical unitary ensemble. 
We will show the universality for the scaling limits of the $k$-the conditional expectation of the overlap in the three regimes, strongly non-unitary, weakly non-unitary, and the singular origin regimes.
\end{abstract}

\maketitle

\section{Introduction}\label{Section1}
Recently, the investigation of overlap, defined by left and right eigenvectors, is one of the hottest topic in random matrix theory.
The overlap was originally introduced by Chalker and Mehling \cite{CM98,CM00}. 
While overlap is a trivial quantity for symmetric classes, such as symmetric/Hermitian/symplectic matrices where left and right eigenvectors match, resulting in the overlap matrix becoming the identity matrix, it becomes a non-trivial object for non-Hermitian matrices. 
Subsequent to a groundbreaking paper by \cite{BD21}, the mathematical analysis of overlap in random matrices has seen significant development in \cite{ATTZ, AFK20, BNST17, BZ18, BSV17, CS, CEHS23, CR22, D21v1, D21v2, D23, EC23, EKY21, F18, FT21, N18, WTF23, Y20}.

The overlap plays a crucial role in describing the stochastic dynamics of eigenvalues for non-Hermitian matrix-valued Brownian motion, \cite{BD21, EKY21, GW18, Y20}. 
In this note, we study the integrability of overlaps for the induced spherical unitary ensemble. 
Specifically, we explore the $k$-th conditional expectation of the overlap for the induced spherical unitary ensemble in the strongly and weakly non-unitary, as well as the singular origin regimes. 
The weakly non-unitary regime plays a role in interpolating between the spherical ensemble (non-normal matrix) and the circular unitary ensemble (unitary matrix). 
This regime has been investigated from the viewpoint of the limiting point process for eigenvalues of non-Hermitian random matrices or two-dimensional Coulomb gases, \cite{AP14, AB23, SC22, SFv3, SS23, FKS97}. 
Building upon the framework established in \cite{ATTZ}, we derive scaling limits for the conditional expectation of both diagonal and off-diagonal overlaps for the induced spherical unitary ensemble in these three regimes. 
As a consequence, we confirm the universality of overlaps.

Before concluding this introduction, let us highlight recent developments in the study of non-Hermitian random matrices and our methodology in this note. The method and motivation in this note, as previously mentioned, are strongly inspired by \cite{ATTZ}. Their exploration of the $k$-th conditional expectation of overlaps for the Ginibre unitary ensemble, where they demonstrate scaling limits and observe algebraic decay for the conditional expectation of both diagonal and off-diagonal overlaps. 
On the other hand, the groundbreaking paper by \cite{BD21} revealed that the diagonal overlap for the Ginibre unitary ensemble converges to the inverse Gamma distribution with a parameter of one in law. 
They also provided correlations for both on diagonal and off diagonal overlaps, as well as diagonal and off-diagonal overlaps at microscopic and mesoscopic scales using a robust probabilistic method.

While our method does not provide results in a distributional sense, we offer explicit quantities for both the conditional expectation of diagonal and off-diagonal overlaps beyond Gaussian-type non-Hermitian matrices. Additionally, for the diagonal overlap, as we will discuss later, we introduce a new point process associated with the weight function deformed by the overlap, distinct from point insertion. Consequently, we believe our results are not only intriguing from the perspective of overlap but also offer valuable insights into point processes. 
We also study the overlaps for the induced spherical unitary ensemble in the weakly non-unitary regime.
It is worth noting that this regime has already been explored for other ensembles, as seen in works such as \cite{F18, FT21,CW24,CFW24}. These studies focused on the full density of on-diagonal overlap for the Ginibre orthogonal \cite{WTF23} and elliptic Ginibre orthogonal ensembles in both strongly and weakly non-unitary regimes. However, their method relies on the supersymmetric method, leaving the results for off-diagonal overlap in both regimes unknown. 
\cite{D21v1, D23} investigated overlaps for the spherical and truncated unitary ensembles, and our results complement their findings for the spherical unitary ensemble. Regarding the overlap of planar symplectic ensembles, we refer only to \cite{AFK20, D21v2}. For studies on the universality of overlap in more general non-Hermitian random matrices, see \cite{CS, CEHS23, EC23}.

\subsubsection*{Organisation in this paper}
In section \ref{Section2}, we introduce induced spherical unitary ensemble and the overlaps.
Also, we will explain three regimes, strongly non-unitary regime, weakly non-unitary regime, and at the singular origin, for eigenvalues of induced spherical unitary ensemble. 
In section \ref{Section3}, we state our main results. 
In section \ref{Section4}, we establish a family of the planar orthogonal polynomials associated with a spherical weight function deformed by the diagonal overlap, and as a consequence, we obtain a finite $N$-kernel. 
However, our finite $N$-kernel is not appropriate for large $N$-analysis. 
To overcome this difficulty, we will simplify the finite $N$-kernel. 
In section \ref{Section5}, based on the simplified finite $N$-kernel, we will compute the asymptotic behavior of that depending on some parameters. 
In section \ref{Section6}, we give some concluding remarks, and we complete this paper. 

\section{Preliminaries}\label{Section2}
In this section, we introduce induced spherical ensemble, and we collect some facts for that. 
For the recent developments of the progress on non-Hermitian random matrices, in particular, integrable non-Hermitian random matrix models, and instructive surveys, we refer to \cite{SF22v1,SF22v2}. 
\subsection{Induced spherical unitary ensemble}
For $n\geq N$ and $L\geq 0$, the induced spherical unitary ensemble is defined by
\begin{equation}
\label{Mdense}
d\P_N(\mathbf{G}_N)=C_N\frac{\det\left(\mathbf{G}_N\mathbf{G}_N^\dagger\right)^{L}}{\det\left(\mathbf{1}_N+\mathbf{G}_N\mathbf{G}_N^\dagger\right)^{n+N+L}}d[\mathbf{G}_N],\quad
C_{N}=\frac{1}{\pi^{N^2}}\prod_{k=1}^N\frac{\Gamma(k)\Gamma(n+N+k)}{\Gamma(L+k)\Gamma(n-N+k)},
\end{equation}
where parameters $n$ and $L$ satisfy $n\geq L$ and $L\geq 0$, and they may depend on $N$.
Induced spherical unitary ensemble can be realized as a random matrix 
$\mathbf{G}_N=\mathbf{U}_N(\mathbf{Y}_N^\dagger\mathbf{Y}_N)^{\frac{1}{2}}$, 
where $\mathbf{U}_N$ is a circular unitary matrix uniformly distributed to the Haar measure on the unitary group, and $\mathbf{Y}_N$ is defined by $\mathbf{Y}_N:=\mathbf{X}\mathbf{A}^{-1/2}$. 
Here, $\mathbf{X}$ is an $(n+L)\times N$ rectangular Ginibre ensemble, whose elements are independent identically distributed standard complex Gaussian random variables, and $\mathbf{A}$ is an $N\times N$ complex Wishart matrix.    
For the detailed construction of induced spherical ensemble and other ensembles such as the induced orthogonal or symplectic ensembles, we refer to \cite{FF11,MA13,MA17}, and see also \cite[Section 1.]{SFv3} for a clear and brief explanation.
The joint probability distribution function of the eigenvalues associated with \eqref{Mdense} is given by 
\begin{equation}
\label{Edensity}
d\P_N(\boldsymbol{z}_{(N)})=\frac{1}{Z_N}\prod_{1\leq j<k\leq N}|z_j-z_k|^2\prod_{j=1}^{N}e^{-NQ(z_j)}dA(z_j),
\end{equation}
for $\boldsymbol{z}_{(N)}=(z_1,\dots,z_N)\in\C^N$, where 
\begin{equation}
\label{Potential}
Q(z)=\frac{n+L+1}{N}\log(1+|z|^2)-\frac{2L}{N}\log|z|.
\end{equation}
Here, $dA(z)$ is the planar Lebesgue measure on the complex plane $\C$. 
Following \cite{SS23}, we introduce the three regimes: 
\textbf{strongly non-unitary, weakly non-unitary}, and \textbf{the singular origin regimes}.
From the basics of logarithmic potential theory \cite{HM13,SFT97}, the empirical measure of the eigenvalues for the induced spherical unitary ensemble tends to be distributed on the spectral droplet 
\begin{equation}
\label{SD}
S=\{z\in\C:r_1\leq |z|\leq r_2\},\quad r_1=\sqrt{\frac{L}{n}},\quad r_2=\sqrt{\frac{N+L}{n-N}},
\end{equation}
with the density
\[
\Delta Q(z)\mathbf{1}_{S}\sim\frac{n+L}{N}\frac{1}{(1+|z|^2)^2}. 
\]
Then, depending on a way to choose the parameters $n,L$, the spectral droplet \eqref{SD} is classified as follows.
\begin{description}
\item[(1) Strongly non-unitary regime]
We set the parameters $L,n$ as follows: for a fixed $a,b \geq 0$, 
\begin{equation}
\label{Para1}
L=aN,\quad n=(b+1)N.
\end{equation}
Then, with these parameters, the inner and outer radii of \eqref{SD} satisfy
\[
r_1=\sqrt{\frac{a}{b+1}}+O(N^{-1}),\quad r_2=\sqrt{\frac{a+1}{b}}+O(N^{-1}),\quad\text{as $N\to\infty$}.
\]
In the strongly non-unitary regime, the width of spectral droplet is $O(1)$. 
\item[(2) Weakly non-unitary regime]
Fora fixed $\rho>0$, we set the parameters as 
\begin{equation}
\label{Para2}
L=\frac{N^2}{\rho^2}-N,\quad n=\frac{N^2}{\rho^2}.
\end{equation}
Then, the inner and outer radii of \eqref{SD} satisfy 
\[
r_1=1-\frac{\rho^2}{2N}+O(N^{-2})\quad
r_2=1+\frac{\rho^2}{2N}+O(N^{-2}),\quad\text{as $N\to\infty$}. 
\]
In the weakly non-unitary regime, the spectral droplet \eqref{SD} accumulates about the unit circle with width $O(N^{-1})$. 
Roughly speaking, the eigenvalue distribution in this regime macroscopically seems to be the circular unitary ensemble. 
\item[(3) The singular origin regime]
For a fixed $b>0$, we set the parameters as 
\begin{equation}
\label{Para3}
L>0, \quad n=(b+1)N.
\end{equation}
Then, the radii of the spectral droplet \eqref{SD} satisfy 
\[
r_1=O(N^{-1}),\quad
r_2=\frac{1}{\sqrt{b}}+O(N^{-1}),\quad \text{as $N\to\infty$}.
\]
\end{description}
\begin{figure}[htbp]
  \begin{minipage}[b]{0.3\linewidth}
    \centering
    \includegraphics[keepaspectratio, scale=0.3]{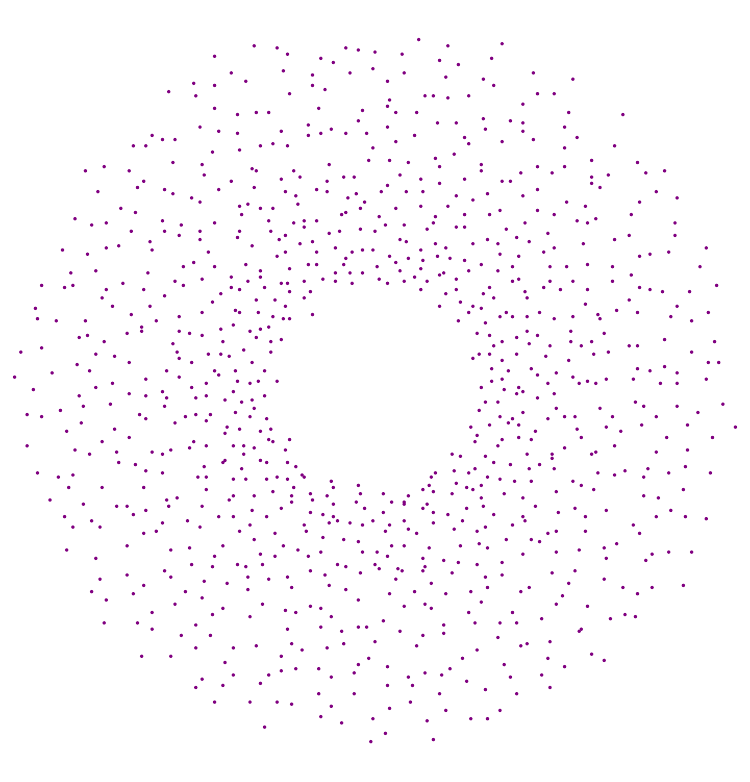}
    \subcaption{Strong non-unitary regime}
  \end{minipage}
  \begin{minipage}[b]{0.3\linewidth}
    \centering
    \includegraphics[keepaspectratio, scale=0.3]{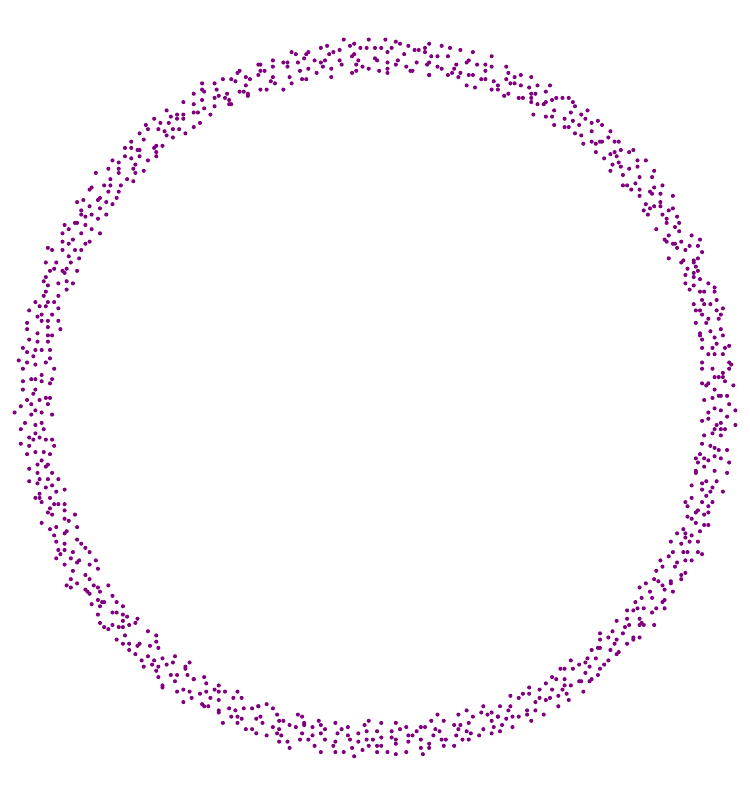}
    \subcaption{Weak non-unitary regime}
  \end{minipage}
    \begin{minipage}[b]{0.3\linewidth}
    \centering
    \includegraphics[keepaspectratio, scale=0.3]{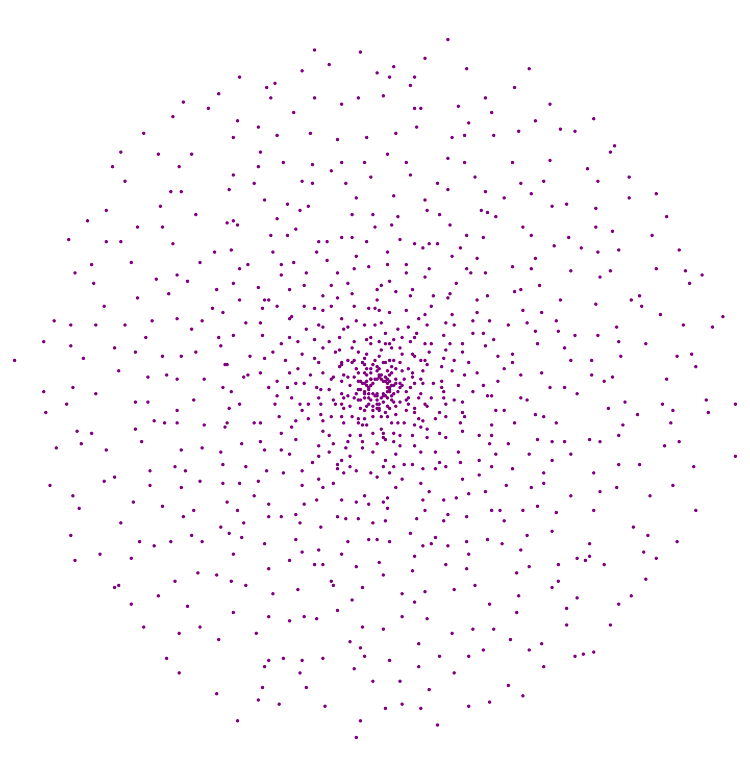}
    \subcaption{Singular origin regime}
  \end{minipage}
    \begin{minipage}[b]{0.3\linewidth}
    \centering
    \includegraphics[keepaspectratio, scale=0.3]{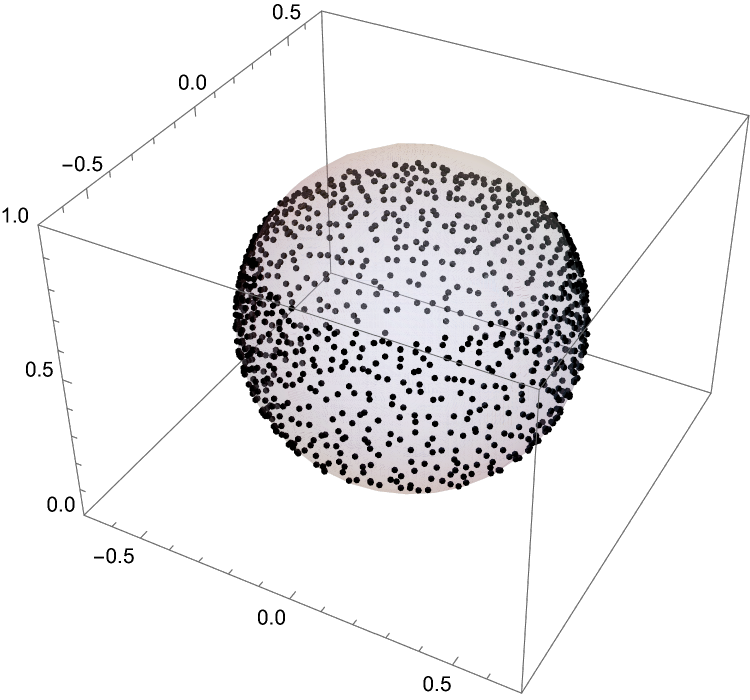}
    \subcaption{Spherical projection in the strongly non-unitary regime}
  \end{minipage}
  \begin{minipage}[b]{0.3\linewidth}
    \centering
    \includegraphics[keepaspectratio, scale=0.3]{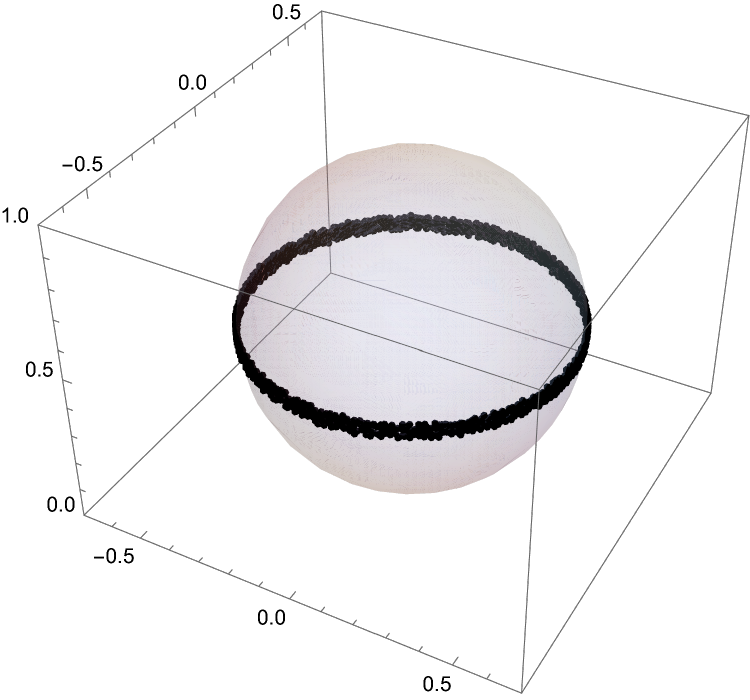}
    \subcaption{Spherical projection in the weakly non-unitary regime}
  \end{minipage}
    \begin{minipage}[b]{0.3\linewidth}
    \centering
    \includegraphics[keepaspectratio, scale=0.3]{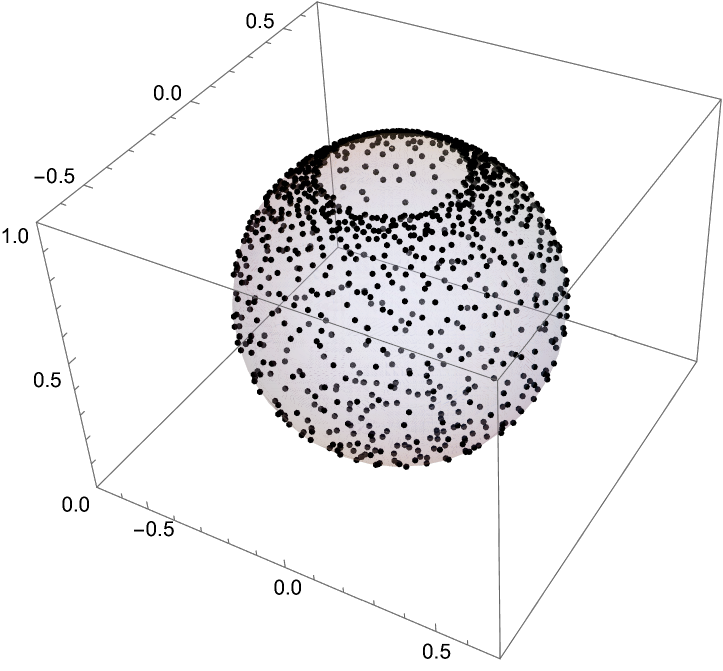}
    \subcaption{Spherical projection at the singular origin regime}
  \end{minipage}
  \caption{(A), (B), and (C) are the plots of the eigenvalues of induced spherical unitary ensemble in each regime on the complex plane. (D), (E), and (F) are the plots of the spherical projection for their eigenvalues.}
\end{figure}
We have discussed the macroscopic pictures for the eigenvalues of the induced spherical unitary ensemble so far. 
Let us mention the point process associated with \eqref{Edensity}. 
It is well-known that the $k$-th correlation function of \eqref{Edensity} forms the determinantal point process, \cite{SF22v1,SF22v2}.
It is given by
\begin{equation}
\label{k correlation function isue}
\mathbf{R}_{N,k}(z_1,\dots,z_k)=\det\bigl( K_N(z_i,z_j) \bigr)_{i,j=1}^{k},
\end{equation}
where $K_N$ is called a correlation kernel given by 
\begin{equation}
\label{RKernel}
K_N(z,w)=e^{-\frac{N}{2}(Q(z)+Q(w))}\sum_{k=0}^{N-1}\frac{\Gamma(n+L+1)}{\Gamma(n-k)\Gamma(k+L+1)}(\overline{z}w)^k. 
\end{equation}
Then, depending the way to choose parameters $L,n$ in the strongly non-unitary, the weakly non-unitary, and the singular origin regimes, the scaling limits of \eqref{RKernel} are classified into the following three classes. 
In the strongly non-unitary regime, we have 
\begin{align*}
&\quad
\lim_{N\to\infty}\frac{1}{N\delta_N(p)}K_N\Bigl(e^{i\theta}\Bigl( p+\mathfrak{s}\frac{\zeta}{\sqrt{N\delta_N(p)}}\Bigr), e^{i\theta}\Bigl( p+\mathfrak{s}\frac{\eta}{\sqrt{N\delta_N(p)}}\Bigr)\Bigr)\\
&=
\begin{cases}
e^{-\frac{1}{2}(|z|^2+|w|^2)+\overline{\zeta}\eta}, & \text{if $p\in\mathrm{int}(S)$, $\theta=0$, and $\mathfrak{s}=1$}, \\
e^{-\frac{1}{2}(|z|^2+|w|^2)+\overline{\zeta}\eta}\erfc(\overline{\zeta}+\eta), 
& \text{if $p\in\partial S$, $\theta\in[0,2\pi)$, and $\mathfrak{s}=1$ (outer edge) or $\mathfrak{-1}$ (inner edge)},
\end{cases}
\end{align*}
uniformly for $\zeta,\eta$ in compact subsets of $\C$. 
In the weakly non-unitary regime, we have 
\[
\lim_{N\to\infty}\frac{1}{N\delta_N(1)}K_N\Bigl(1+\frac{\zeta}{\sqrt{N\delta_N(1)}}, 1+\frac{\eta}{\sqrt{N\delta_N(1)}}\Bigr)
=
e^{-\frac{1}{2}(|z|^2+|w|^2)+\overline{\zeta}\eta}L_{\rho}(\overline{\zeta}+\eta),
\]
uniformly for $\zeta,\eta$ in compact subsets of $\C$. For the definition of $L_{\rho}$, see \eqref{WeakL1}. 
Similarly, at the singular origin regime, the limiting reproducing kernel is characterized by the two-parametric Mittag-Leffler function \eqref{ThmSingular1}, see \cite{AKS21}. 
For the detailed explanation, we refer to \cite{SFv3,FI12}.

\subsection{Overlap for the induced spherical ensemble}
In this subsection, we briefly discuss an overlap defined by left and right eigenvectors for a non-normal matrix, in particular, we focus on the induced spherical ensemble following \cite{D21v1}. In \cite{D21v1}, the case of $L=0$ was only discussed. 
For a given non-normal matrix $\mathbf{G}\in\mathbb{M}_N(\C)$ with simple spectrums $\boldsymbol{z}_{(N)}=(z_1,z_2,...,z_N)\in\C^{N}$, 
we denote a right and left eigenvector associated with an eigenvalue $z_k$ by $R_k$ and $L_k$, respectively, which are defined by 
\[
\mathbf{G}R_k=z_kR_k,\quad 
L_k^{\mathsf{t}}\mathbf{G}=z_kL_k^{\mathsf{t}},
\]
satisfying the bi-orthogonal condition 
\[
L_j^{\mathsf{t}}R_k=\delta_{j,k}.
\]
Then, we define an overlap matrix $\mathcal{O}$ with entries 
\[
\mathcal{O}_{j,k}=(L_j^{\dagger}L_k)(R_j^{\dagger}R_k).
\]
For the definition of the overlap, some remarks should be mentioned. 
The overlap is invariant under the scaling $L_j\mapsto \frac{1}{c}L_j$ and $R_k\mapsto c R_k$ for a constant $c\neq0$. 
Also, the overlap is unchanged by a unitary change of basis.
By the exchangeability of the eigenvalues, we focus on the diagonal overlap $\mathcal{O}_{1,1}$ and the off-diagonal overlap $\mathcal{O}_{1,2}$. 
Similar to the proof of the case for $L=0$ in \cite{D21v1}, we can show the followings:
\begin{prop}\label{prop_O11}
Conditionally on $\boldsymbol{z}_{(N)}=\boldsymbol{\lambda}_{(N)}$, the on-diagonal overlap $\mathcal{O}_{1,1}$ of the induced spherical unitary ensemble is distributed as 
\[
\mathcal{O}_{1,1}\overset{d}{=}\prod_{k=2}^{N}\Bigl(1+\frac{(1+|\lambda_1|^2)(1+|\lambda_k|^2)}{|\lambda_1-\lambda_k|^2}
X_{n+L}^{(k)}
\Bigr),
\]
where $X_m^{(k)}$ is i.i.d. distributed to a real random variable with density 
\[
\frac{m+1}{(1+x)^{m+2}}\mathbf{1}_{\R_+}.
\]
In particular, the quenched expectation is given by
\[
\E_N[\mathcal{O}_{1,1}|\boldsymbol{z}_{(N)}=\boldsymbol{\lambda}_{(N)}]
=
\prod_{k=2}^{N}\Bigl(1+\frac{(1+|\lambda_1|^2)(1+|\lambda_k|^2)}{(n+L)|\lambda_1-\lambda_k|^2}\Bigr), 
\]
where $\E_N$ denotes the expectation with respect to \eqref{Mdense}. 
\end{prop} 
\begin{prop}\label{prop_O12}
The quenched expectation of the off-diagonal overlap $\mathcal{O}_{1,2}$ for the induced spherical unitary ensemble is given by
\[
\E_N[\mathcal{O}_{1,2}|\boldsymbol{z}_{(N)}=\boldsymbol{\lambda}_{(N)}]
=
-\frac{1}{(n+L)|\lambda_1-\lambda_2|^2}
\prod_{k=3}^{N}\Bigl(1+\frac{(1+\lambda_1\overline{\lambda_2})(1+|\lambda_k|^2)}{(n+L)(\lambda_1-\lambda_k)(\overline{\lambda_2-\lambda_k})}\Bigr).
\]
\end{prop} 
Proposition~\ref{prop_O11} corresponds to \cite[Theorem 2.6]{D21v1}, and Proposition~\ref{prop_O12} corresponds to \cite[Theorem 2.11]{D21v1}. 
Our aim in this note is to study the determinantal structure of the $k$-th conditional expectation of the overlaps for the induced spherical unitary ensemble, that is, for $\boldsymbol{z}_{(k)}=(z_1,z_2,...,z_k)\in\C^k$, we study 
\begin{align}
\begin{split}
\label{D11}
&\quad D_{1,1}^{(N,k)}(\boldsymbol{z}_{(k)})
\\
&:=\frac{N!}{(N-k)!}\frac{1}{Z_{N}}\int_{\C^{N-k}}
\prod_{j=2}^N\Bigl(1+\frac{(1+|z_1|^2)(1+|z_j|^2)}{(n+L)|z_1-z_j|^2}\Bigr)
|\Delta_N(\boldsymbol{z}_{(N)})|^2\prod_{j=1}^{N}e^{-NQ(z_j)}\prod_{j=k+1}^{N}dA(z_j),
\end{split}
\end{align}
and 
\begin{align}
\label{D12}
D_{1,2}^{(N,k)}(\boldsymbol{z}_{(k)})
&:=
-\frac{N!}{(N-k)!}\frac{1}{Z_{N}}\int_{\C^{N-k}}\frac{1}{(n+L)|z_1-z_2|^2}
\\
&\quad\times
\prod_{j=3}^N\Bigl(1+\frac{(1+z_1\overline{z_2})(1+|z_j|^2)}{(n+L)(z_1-z_j)(\overline{z_2-z_j})}\Bigr)
|\Delta_N(\boldsymbol{z}_{(N)})|^2\prod_{j=1}^{N}e^{-NQ(z_j)}\prod_{j=k+1}^{N}dA(z_j),
\nonumber
\end{align}
where, $\Delta_N(\boldsymbol{z}_{(N)})$ is the Vandemonde determinant for $\boldsymbol{z}_{(N)}\in\C^N$.
In order to show the determinantal structure, i.e., to construct the correlation kernel for \eqref{D11} and \eqref{D12}, we introduce a new weight function on $\C$:
\begin{equation}
\label{Weight}
\omega(z,\overline{z}|u,v)
=
\Bigl((z-u)(\overline{z}-v)+\frac{(1+uv)(1+z\overline{z})}{n+L}\Bigr)
e^{-NQ(z)},\quad \text{for $z,u,v,\in\C$},
\end{equation}
with respect to the planar Lebesgue measure $dA(z)$.
Then, \eqref{D11} and \eqref{D12} can be written as 
\begin{align}
\begin{split}
\label{D11v2}
&\quad D_{1,1}^{(N,k)}(\boldsymbol{z}_{(k)})
\\
&=
\frac{N!}{(N-k)!}\frac{1}{Z_N}\frac{|z_1|^{2L}}{(1+|z_1|^2)^{n+L+1}}\int_{\C^{N-k}}|\Delta_{N-1}(z_2,z_3,...,z_N)|^2\prod_{j=2}^{N}\omega(z_j,\overline{z_j}|z_1,\overline{z_1})\prod_{j=k+1}^{N}dA(z_j),
\end{split}
\end{align}
and 
\begin{align}
\begin{split}
\label{D12v2}
&\quad D_{1,2}^{(N,k)}(\boldsymbol{z}_{(k)})
\\
&=
-
\frac{N!}{(N-k)!}\frac{1}{Z_N}\frac{|z_1|^{2L}|z_2|^{2L}}{(1+|z_1|^2)^{n+L+1}(1+|z_2|^2)^{n+L+1}}
\\
&\times
\int_{\C^{N-k}}\Delta_{N-1}(z_1,z_3,...,z_N)\overline{\Delta_{N-1}(z_1,z_3,...,z_N)}\prod_{j=3}^{N}\omega(z_j,\overline{z_j}|z_1,\overline{z_2})\prod_{j=k+1}^{N}dA(z_j).
\end{split}
\end{align}
Let $T$ be the transposition operator acting on functions $f$ on $\C^{2k}$ with $k\geq 2$, and we treat a complex variable $z$ and its complex conjugate $\overline{z}$ as independent. 
The action of $T$ is defined by exchanging the variables $\overline{z}_1\leftrightarrow\overline{z}_2$:
\begin{equation}
    \label{action of T}
    Tf(z_1,\overline{z}_1,z_2,\overline{z}_2,\cdots)
    =
    f(z_1,\overline{z}_2,z_2,\overline{z}_1,\cdots).
\end{equation}
In particular, $T$ is the transposition operator leaving the remaining variables $z_3,\overline{z}_3,\dots,z_k,\overline{z}_k$ unchanged if present. 
Then, the relationship between \eqref{D11v2} and \eqref{D12v2} can be associated by the following lemma~\ref{Decoupling}, which generalise \cite[Lemma 1]{ATTZ} for the GinUE and see also the recent progress \cite[Lemma 2.2.]{ABN24} for GinSE. 
\begin{lem}
\label{Decoupling}
We have 
\begin{equation}
D_{1,2}^{(N,k)}(\boldsymbol{z}_{(k)})
=
-\frac{|1+z_1\overline{z_2}|^{2(n+L+1)}}{(1+|z_1|^2)^{n+L+1}(1+|z_2|^2)^{n+L+1}}
\frac{TD_{1,1}^{(N,k)}(\boldsymbol{z}_{(k)})}{\frac{|1+z_1\overline{z_2}|^2}{n+L}-|z_1-z_2|^2}.
\end{equation}
\end{lem}
Since the proof of Lemma \ref{Decoupling} is almost same as Ginibre unitary ensemble case \cite{ATTZ}, 
the detailed proof is left to interested readers.  
This lemma allows us to focus on the on-diagonal overlap case. 
Indeed, all proofs in this note will be done for the on-diagonal overlap case, 
and the results for the off-diagonal overlap case will be obtained as a simple application of Lemma \ref{Decoupling}. 
Here, let us mention a reason why we study such quantities \eqref{D11v2} and \eqref{D12v2}. 
We can see the information about the diagonal and the off-diagonal overlaps from \eqref{D11} and \eqref{D12}. 
To see that, by the definition of the conditional expectation, 
we have 
\[
\E_N[\mathcal{O}_{1,1}|\lambda_1=z_1]
=
\frac{\E_{N-1}\Bigl[
\prod_{k=2}^{N}\Bigl(|z_1-\lambda_k|^2+\frac{(1+|z_1|^2)(1+|\lambda_k|^2)}{n+L}\Bigr)
\Bigr]}{\E_{N-1}\Bigl[\prod_{j=2}^N|\lambda_j-z_1|^2\Bigr]}. 
\]
Note that by \cite[Theorem 1]{AV03}, we have 
\begin{equation}
    \E_{N-1}\Bigl[
    \prod_{j=2}^N|\lambda_j-z_1|^2
    \Bigr]
    =
    h_{N-1}^{(\mathrm{ISUE})}\mathcal{K}_N^{(n,L)}(z,w),
\end{equation}
where 
\begin{equation}
    \mathcal{K}_N^{(n,L)}(z,w):=\sum_{k=0}^{N-1}\frac{(z\overline{w})^k}{h_k^{(\mathrm{ISUE})}}, \quad
    h_{k}^{(\mathrm{ISUE})}:=\frac{\Gamma(n-k)\Gamma(k+L+1)}{\Gamma(n+L+1)}. 
\end{equation}
Hence, we have 
\begin{equation}
\label{On conditional expectation 1}
    \E_N[\mathcal{O}_{1,1}|\lambda_1=z_1]
    =
    \frac{N!\prod_{j=0}^{N-2}h_j(z_1,\overline{z}_1)}{Z_{N}\cK_{N}(z_1,\overline{z}_1)}
    =
    \frac{D_{1,1}^{(N,1)}(z_1)}{\mathbf{R}_{N,1}(z_1)},
\end{equation}
where $h_j\equiv h_j(z_1,\overline{z}_1)$ is given by \eqref{LDU result} and $Z_N$ is the normalization constant given by \eqref{Mdense}. Here, $\mathbf{R}_{N,1}$ is the one-point density for the induced spherical unitary ensemble given by \eqref{k correlation function isue}. 
For the off-diagonal overlap, similar to the on-diagonal overlap, we have
\begin{align*}
&\quad\E_N[\mathcal{O}_{1,2}|\lambda_1=z_1,\lambda_2=z_2]
\\
&=
-\frac{1}{(n+L)|\lambda_1-\lambda_2|^2}
\frac{\E_{N-2}\Bigl[
\prod_{k=3}^{N}
(\overline{z}_1-\overline{\lambda}_k)(z_2-\lambda_k)
\Bigl((z-\lambda_k)(\overline{z_2-\lambda_k})+\frac{(1+z_1\overline{z_2})(1+|\lambda_k|^2)}{n+L}\Bigr)
\Bigr]}{\E_{N-2}\Bigl[
\prod_{k=3}^N|z_1-\lambda_k|^2|z_2-\lambda_k|^2
\Bigr]}.
\end{align*}
Then, by Lemma~\ref{Decoupling} and \cite[Theorem 1]{AV03}, we have 
\begin{equation}
\label{Off conditional expectation 1}
    \E_N[\mathcal{O}_{1,2}|\lambda_1=z_1,\lambda_2=z_2]
    =
    \frac{D_{1,2}^{(N,2)}(z_1,z_2)}{Z_{N}^{(n,L)}(z_1,z_2)},
\end{equation}
where 
\begin{equation}
    Z_{N}^{(n,L)}(z_1,z_2)
    :=
    \widehat{\omega}^{(n,L)}(\overline{z}_1,z_1)
\widehat{\omega}^{(n,L)}(\overline{z}_2,z_2)
\det
\begin{bmatrix}
\mathcal{K}_{N}^{(n,L)}(z_1,z_1) & \mathcal{K}_{N}^{(n,L)}(z_1,z_2) \\
\mathcal{K}_{N}^{(n,L)}(z_2,z_1) & \mathcal{K}_{N}^{(n,L)}(z_2,z_2)
\end{bmatrix},
\end{equation}
where $\widehat{\omega}^{(n,L)}$ is defined by \eqref{Homega}. 
Hence, once we know the scaling limits of \eqref{D11} and \eqref{D12}, we could get the scaling limits of the conditional expectations of the on- and off-diagonal overlaps. 

Finally, as we already mentioned, overlaps for non-Hermitian random matrices can be applied to the study of the dynamical eigenvalues processes of non-Hermitian matrix-valued Brownian motions.
In \cite[page 2]{ATTZ}, the motivation to study \eqref{D11}, \eqref{D12} was mentioned from the perspective of the relationship between \eqref{D11}, \eqref{D12}, and the correlation of the dynamical eigenvalues processes of the non-Hermitian matrix-valued Brownian motion. 
It would be interesting to study the construction of dynamical eigenvalues processes of the dynamical induced spherical unitary ensemble whose two matrices to define it are non-Hermitian matrix-valued Brownian motions and to find the similar relationship mentioned in \cite{ATTZ} for the present setting, but it seems to be out of reach in this paper.

\section{Main results}\label{Section3}
In this section, we state our main results in this note. 
We write 
\begin{equation}
\label{GNL1}
g_m^{(n,L)}(x)=\frac{x-L/n}{x} \sum_{k=0}^{m}(m+1-k)\binom{L+n}{L+k}x^k + \frac{L(m+1)}{nx} \binom{L+n}{L}. 
\end{equation}
We also write 
\begin{equation}
\label{QNL}
q_N^{(n,L)}(x)=
\sum _{k=0}^N \frac{\Gamma (L+n+1)}{\Gamma (k+L+1) \Gamma (n+1-k)}x^k, 
\end{equation}
and 
\begin{equation}
\label{HQNL}
\widehat{q}_N^{(n,L)}(y|x)=
q_N^{(n,L)}(y)+\frac{1}{nx-L} \frac{\Gamma(L+n+1)}{\Gamma(n+1)\Gamma(L)}. 
\end{equation}
Let us denote 
\begin{align}
\begin{split}
\label{GNL2}
\widehat{g}_N^{(n,L)}(x)
=
(L+N+1)\widehat{q}_{N+1}^{(n,L)}(x|x)-x(n-N-1)\widehat{q}_{N}^{(n,L)}(x|x).
\end{split}
\end{align}
We have the relationship between \eqref{GNL1} and \eqref{GNL2}:
\begin{equation}
g_N^{(n,L)}(x)=\frac{(x-L/n)}{x(1+x)}\widehat{g}_N^{(n,L)}(x). 
\end{equation}
Let us also denote 
\begin{equation}
\label{Homega}
\widehat{\omega}^{(n,L)}(z,w)=\frac{(zw)^{L}}{((1+z\overline{z})(1+w\overline{w}))^{\frac{n+L+1}{2}}}
\end{equation}
and
\begin{align}
\begin{split}
\label{vapi}
\varpi^{(n,L)}(z,\overline{z},w,\overline{w}|\lambda,\overline{\lambda})
&=
\sqrt{
|z-\lambda|^2+\frac{(1+\lambda\overline{\lambda})(1+z\overline{z})}{n+L}
}
\sqrt{
|w-\lambda|^2+\frac{(1+\lambda\overline{\lambda})(1+w\overline{w})}{n+L}
}.
\end{split}
\end{align}
First, we construct a finite $N$-kernel associated with \eqref{Weight}. 
\begin{prop}\label{Prop_FiniteN}
We write 
\begin{align}
\begin{split}
\label{FiniteN1}
G_{N}^{(n,L)}(x|y,z)&=\sum_{s,t=0}^{N-1}g_s^{(n,L)}(x)y^sg_t^{(n,L)}(x)z^t
\\
&\times
\sum_{k=\max\{s,t\}}^{N-1}
\frac{\Gamma(n+L+1)}{\Gamma(k+L+2)\Gamma(n-k-1)}
\frac{x^k}{g_{k+1}^{(n,L)}(x)g_{k}^{(n,L)}(x)}.
\end{split}
\end{align}
Then, for any $2\leq k\leq N$, we have 
\begin{equation}
\label{D11ver1}
D_{1,1}^{(N,k)}(\boldsymbol{z}_k)
=
\frac{n\left(z_1\zbar_1-\frac{L}{n}\right)}{z_1\zbar_1(1+z_1\zbar_1)}
\widehat{g}_{N-1}^{(n,L)}(z_1\zbar_1)
\widehat{\omega}^{(n,L)}(\overline{z_1},z_1)
\underset{2\leq i,j\leq N}{\det}\Bigl(K_{1,1}^{(N-1)}(z_i,\zbar_i,z_j,\zbar_j|z_1,\zbar_1)\Bigr),
\end{equation}
where $K_{1,1}^{(N)}$ is the correlation kernel given by
\begin{equation}
\label{onK11}
K_{1,1}^{(N)}(z,\zbar,w,\wbar|\lambda,\bar{\lambda})
=
\cK_{1,1}^{(N)}(\zbar,w|\lambda,\bar{\lambda})
\varpi^{(n,L)}(z,\overline{z},w,\overline{w}|\lambda,\overline{\lambda})
e^{-\frac{N}{2}(Q(z)+Q(w))}.
\end{equation}
Here, $\cK_{1,1}^{(N)}$ is the polynomial kernel given by
\begin{equation}
\label{FiniteN2}
\cK_{1,1}^{(N)}(\zbar,w|\lambda,\lambda)=G_{N}^{(n,L)}\Bigl(\lambda\bar{\lambda}|\frac{\bar{z}}{\bar{\lambda}},\frac{w}{\lambda}\Bigr).
\end{equation}
Furthermore, for $k\geq 2$, we have 
\begin{align}
\begin{split}
\label{D12ver1}
D_{1,2}^{(N,k)}(\boldsymbol{z}_k)&=
-
\frac{n(z_1\overline{z_2}-\frac{L}{n})}{z_1\overline{z_2}(1+z_1\overline{z_2})}
\widehat{g}_{N-1}^{(n,L)}(z_1\overline{z_2})\widehat{\omega}^{(n,L)}(\overline{z_2},z_1)
\\
&\times
\mathcal{K}^{(N-1)}(\overline{z_1},z_2|z_1,\overline{z_2})
\widehat{\omega}^{(n,L)}(\overline{z_1},z_2)
\underset{3\leq i,j\leq N}{\det}\Bigl(K_{1,2}^{(N-1)}(z_i,\zbar_i,z_j,\zbar_j|z_1,\zbar_1)\Bigr),
\end{split}
\end{align}
where
\begin{align}
\begin{split}
&\quad
K_{1,2}^{(N)}(z,\zbar,w,\wbar|u,\bar{u},v,\bar{v})
\\
&
=
\frac{\varpi^{(n,L)}(z,\overline{z},w,\overline{w}|u,\overline{v})
e^{-\frac{N}{2}(Q(z)+Q(w))}}{\cK_{1,1}^{(N)}(\bar{u},v|u,\bar{v})}
\det\begin{pmatrix}
\cK_{1,1}^{(N)}(\bar{u},v|u,\bar{v}) & \cK_{1,1}^{(N)}(\bar{u},w|u,\bar{v}) \\
\cK_{1,1}^{(N)}(\bar{z},v|u,\bar{v}) & \cK_{1,1}^{(N)}(\bar{z},w|u,\bar{v})
\end{pmatrix}.
\end{split}
\end{align}
\end{prop}
In particular, since we get $K_{12}^{(N)}$ from $K_{11}^{(N)}$ via Lemma~\ref{Decoupling}, 
we mainly focus on the diagonal overlap case. 
To directly analyze the finite $N$-kernel \eqref{onK11} in Proposition~\ref{Prop_FiniteN} in a large $N$-limit is difficult. 
Indeed, \eqref{FiniteN1} constitutes many summations and complicated terms in the denominator. 
Therefore, in order to compute the scaling limits of \eqref{onK11}, we need to simplify \eqref{onK11} in terms of more simple functions.
The below result is the building block in this note. 
\begin{thm}\label{Thm_Simplification}
We define
\begin{align}
\begin{split}
\label{PartHN}
\mathfrak{H}_N^{(n,L)}(\zbar,w,\lambda)
&=
\frac{(1+\lambda\overline{\lambda})}{(\overline{z}-\overline{\lambda})^2(w-\lambda)^2\widehat{g}_N^{(n,L)}(\lambda\overline{\lambda})\widehat{\omega}^{(n,L)}(\overline{\lambda},\lambda)}
\\
&\times \Bigl( (N+L+1)Q_{N+1}^{(n,L)}(\overline{z},w,\lambda) -(n-N-1)\lambda\overline{\lambda} Q_{N}^{(n,L)}(\overline{z},w,\lambda) \Bigr), 
\end{split}
\end{align}
where 
\begin{align}
\begin{split}
\label{PartQN}
Q_{N}^{(n,L)}(\overline{z},w,\lambda)
&=
\widehat{q}_{N}^{(n,L)}(\overline{\lambda}w|\lambda\overline{\lambda})
\widehat{\omega}^{(n,L)}(\overline{\lambda},w)
\widehat{q}_{N}^{(n,L)}(\overline{z}\lambda|\lambda\overline{\lambda})
\widehat{\omega}^{(n,L)}(\overline{z},\lambda)
\\
&
-\Bigl( 1-\frac{(\overline{z}-\overline{\lambda})(w-\lambda)}{1+\lambda\overline{\lambda}}\Bigr)\widehat{q}_{N}^{(n,L)}(\lambda\overline{\lambda}|\lambda\overline{\lambda})
\widehat{\omega}^{(n,L)}(\overline{\lambda},\lambda)
\widehat{q}_{N}^{(n,L)}(\overline{z}w|\lambda\overline{\lambda})
\widehat{\omega}^{(n,L)}(\overline{z},w).
\end{split}
\end{align}
We also define
\begin{equation}
\label{PartFN}
\mathfrak{F}_N^{(n,L)}(\overline{z},w,\lambda)
=
\mathfrak{I}_N^{(n,L)}(\overline{z},w,\lambda)
+
\mathfrak{II}_N^{(n,L)}(\overline{z},w,\lambda)
+
\mathfrak{III}_N^{(n,L)}(\overline{z},w,\lambda),
\end{equation}
where 
\begin{align}
\begin{split}
\label{Part1N}
&\quad \mathfrak{I}_N^{(n,L)}(\overline{z},w,\lambda)
\\
&=
\frac{
(L+N)
\widehat{q}_{N}^{(n,L)}(\overline{z}w| \overline{\lambda}\lambda)\widehat{\omega}^{(n,L)}(\overline{z},w)
+(n-N-1-\overline{z}w)
\widehat{q}_{N-1}^{(n,L)}(\overline{z}w| \overline{\lambda}\lambda)\widehat{\omega}^{(n,L)}(\overline{z},w)
}
{(1+\overline{z}w)(\overline{z}-\overline{\lambda})(w-\lambda)} 
\\
&
-
\frac{n(1+\lambda\overline{\lambda})}{n\overline{\lambda}\lambda-L}
\frac{\Gamma(L+n+1)}{\Gamma(n+1)\Gamma(L)}
\frac{\widehat{\omega}^{(n,L)}(\overline{z},w)}{(1+\overline{z}w)(\overline{z}-\overline{\lambda})(w-\lambda)}, 
\end{split}
\end{align}

\begin{align}
\begin{split}
\label{Part2N}
\mathfrak{II}_N^{(n,L)}(\overline{z},w,\lambda)
&=
-\frac{(n-N-1)\Gamma(L+n+1)}{\Gamma(L+N+1)\Gamma(n-N)}
\frac{(\overline{z}w)^N\widehat{\omega}^{(n,L)}(\overline{z},w)}{(\overline{z}-\overline{\lambda})(w-\lambda)\widehat{g}_N^{(n,L)}(\lambda\overline{\lambda})\widehat{\omega}^{(n,L)}(\overline{\lambda},\lambda)}
\\
&\times
\Bigl( 
\widehat{q}_{N-1}^{(n,L)}(\lambda\overline{\lambda}|\lambda\overline{\lambda})
\widehat{\omega}^{(n,L)}(\overline{\lambda},\lambda)
-\frac{\widehat{q}_{N}^{(n,L)}(\lambda\overline{\lambda}|\lambda\overline{\lambda})
\widehat{\omega}^{(n,L)}(\overline{\lambda},\lambda)}{n-N}
\Bigr),
\end{split}
\end{align}
and
\begin{align}
\begin{split}
\label{Part3N}
\mathfrak{III}_N^{(n,L)}(\overline{z},w,\lambda)
&=
\frac{(n-N-1)\Gamma(L+n+1)}{\Gamma(L+N+1)\Gamma(n-N)}
\frac{(\overline{z}w)^N\widehat{\omega}^{(n,L)}(\overline{z},w)}{(\overline{z}-\overline{\lambda})(w-\lambda)\widehat{g}_N^{(n,L)}(\lambda\overline{\lambda})\widehat{\omega}^{(n,L)}(\overline{\lambda},\lambda)}
\\
&\times
\Bigl( 
\frac{L+N+1}{n-N}
\widehat{q}_{N}^{(n,L)}(\lambda\overline{\lambda}|\lambda\overline{\lambda})
\widehat{\omega}^{(n,L)}(\overline{\lambda},\lambda)
-\frac{\overline{z}w\widehat{q}_{N+1}^{(n,L)}(\lambda\overline{\lambda}|\lambda\overline{\lambda})
\widehat{\omega}^{(n,L)}(\overline{\lambda},\lambda)}{n-N-1}
\Bigr).
\end{split}
\end{align}
Then, we can express \eqref{onK11} as 
\begin{equation}
K_{1,1}^{(N)}(z,\overline{z},w,\overline{w}|\lambda,\overline{\lambda})
=
C_{K}(z,w)
\bigl(
\mathfrak{H}_N^{(n,L)}(\overline{z},w,\lambda)
+
\mathfrak{F}_N^{(n,L)}(\overline{z},w,\lambda)
\bigr)
\varpi^{(n,L)}(z,\overline{z},w,\overline{w}|\lambda,\overline{\lambda}),
\end{equation}
where 
\[
C_{K}(z,w)=\frac{\sqrt{\widehat{\omega}^{(n,L)}(\overline{z},z)\widehat{\omega}^{(n,L)}(\overline{w},w)}}{\widehat{\omega}^{(n,L)}(\overline{z},w)}. 
\]
\end{thm}
\begin{rem}
As we will see later, the remainder terms \eqref{Part2N} and \eqref{Part3N} do not play essential roles in a large $N$-limit. 
Indeed, the main contributions come from \eqref{PartHN} and \eqref{Part1N}. 
\end{rem}
Now, we are ready to show the scaling limits based on Theorem~\ref{Thm_Simplification}. 
For $p\in\mathrm{clo}(S)$, we set 
\begin{equation}
\label{DensQ}
\delta_N(p)=\frac{n+L+1}{N}\frac{1}{(1+|p|^2)^2},
\end{equation}
which is the density at a point of \eqref{SD}. 
Let us denote 
\begin{equation}
z_j=e^{i\theta}\Bigl(p+\mathfrak{s}\frac{\zeta_j}{\sqrt{N\delta_N(p)}}\Bigr),
\end{equation}
for $\zeta_j$ in compact subsets of $\C$ for $j=1,2,...,k$, where $\mathfrak{s}=-1$ if the inner edge case in the strongly non-unitary regime, and otherwise $\mathfrak{s}=1$. 
We write 
\begin{equation}
\label{LimitW1}
\varpi(\zeta,\overline{\zeta},\eta,\overline{\eta}|\chi,\overline{\chi})
=
\sqrt{(1+(\overline{\zeta}-\overline{\chi})(\zeta-\eta))( 1+(\overline{\eta}-\overline{\chi})(\eta-\eta))},
\end{equation}
and
\begin{equation}
K_{1,1}^{(\ast)}(\zeta,\overline{\zeta},\eta,\overline{\eta}|\chi,\overline{\chi})
=
\cK_{1,1}^{(\ast)}(\overline{\zeta},\eta|\chi,\overline{\chi})
\omega^{(\ast)}(\overline{\zeta},\eta|\chi,\overline{\chi}), \text{for $\zeta,\eta,\chi\in\C$}.
\end{equation}
Here, the symbols $\mathbf{b,e,w,s}$ are assigned to $\ast$ for the bulk and edge cases in strongly non-unitary regime, in the weakly non-unitary regime, and at the singular origin regime, respectively. 
In the same rule, we also write 
\begin{equation}
K_{1,2}^{(\ast)}(\zeta,\overline{\zeta},\eta,\overline{\eta}|ui,\overline{v})
=
\frac{\omega^{(\ast)}(\overline{\zeta},\eta|u,\overline{v})}{\cK_{1,1}^{(\ast)}(\bar{u},v|u,\bar{v})}
\det\begin{pmatrix}
\cK_{1,1}^{(\ast)}(\bar{u},v|u,\bar{v}) & \cK_{1,1}^{(\ast)}(\bar{u},\eta|u,\bar{v}) \\
\cK_{1,1}^{(\ast)}(\bar{\zeta},v|u,\bar{v}) & \cK_{1,1}^{(\ast)}(\bar{\zeta},\eta|u,\bar{v})
\end{pmatrix},\quad\text{for $\zeta,\eta,u,v\in\C^4$}.
\end{equation}
The below is our main result in this note. 
\begin{thm}\label{MainResult}
We write 
\begin{equation}
\mathfrak{D}_{1,m}^{(N,k)}(\boldsymbol{\zeta}_{(k)})=\frac{1}{(N\delta_N(p))^k}D_{1,m}^{(N,k)}(\boldsymbol{z}_{(k)}),\quad \text{for $m=1,2$}. 
\end{equation}
For $m=1,2$, we set 
\begin{align}
\begin{split}
&\quad
\lim_{N\to\infty} \frac{1}{N^q} \mathfrak{D}_{1,m}^{(N,k)}(\boldsymbol{\zeta}_{(k)})
\\
&
= 
\begin{cases}
\mathfrak{D}_{1,m}^{(k,\mathbf{b})}(\boldsymbol{\zeta}_{(k)}) 
 & \text{if  $ p\in\mathrm{int}(S)$ and strongly non-unitary regime with $q=1$},  \\
 \mathfrak{D}_{1,m}^{(k,\mathbf{e})}(\boldsymbol{\zeta}_{(k)}) 
 & \text{if  $ p\in\partial(S)$ and strongly non-unitary regime with $q=1/2$},  \\
 \mathfrak{D}_{1,m}^{(k,\mathbf{w})}(\boldsymbol{\zeta}_{(k)}) 
 & \text{if  $ p=1$ and weakly non-unitary regime with $q=0$},  \\
 \mathfrak{D}_{1,m}^{(k,\mathbf{s})}(\boldsymbol{\zeta}_{(k)}) 
 & \text{if  $p=0$ and at the singular regime with $q=1$}. 
\end{cases}
\end{split}
\end{align}
\begin{itemize}
\item[(I)] \textbf{\textup{(Strongly non-unitary regime: bulk case)}} 
We write 
\begin{equation}
\cK_{1,1}^{(\mathbf{b})}(\overline{\zeta},\eta|\chi,\overline{\chi})
=
\frac{d}{dx}\Bigl(\frac{e^x-1}{x}\Bigr)\Bigr|_{x=(\overline{\zeta}-\overline{\chi})(\eta-\chi)},
\end{equation}
and 
\begin{equation}
\label{BulkWeight}
\omega^{(\mathbf{b})}(\overline{\zeta},\eta|\chi,\overline{\chi})
=
\varpi(\zeta,\overline{\zeta},\eta,\overline{\eta}|\chi,\overline{\chi})e^{-(\overline{\zeta}-\overline{\chi})(\eta-\chi)}. 
\end{equation}
Then, we have
\begin{equation}
\label{D11BULK}
\mathfrak{D}_{1,1}^{(k,\mathbf{b})}(\boldsymbol{\zeta}_{(k)}) 
=
\frac{b(b+1)}{a+b+1}
\frac{\Bigl(|p|^2-\frac{a}{b+1} \Bigr)\Bigl(\frac{a+1}{b}-|p|^2 \Bigr)}{|p|^2}
\overset{k}{\underset{i,j=2}{\det}} 
\bigl(K_{1,1}^{(\mathbf{b})}(\zeta_i,\overline{\zeta_i},\zeta_j,\overline{\zeta_j}|\zeta_1,\overline{\zeta_1}) \bigr),
\end{equation}
uniformly for $\zeta_k$ in compact subsets of $\C$. 
Moreover, we have 
\begin{align}
\begin{split}
\label{D12BULK}
 \mathfrak{D}_{1,2}^{(k,\mathbf{b})}(\boldsymbol{\zeta}_{(k)})
&=
-\frac{b(b+1)}{a+b+1}\frac{\Bigl( |p|^2-\frac{a}{b+1} \Bigr)\Bigl( \frac{a+1}{b}-|p|^2 \Bigr)}{|p|^2}
\mathcal{K}^{(\mathbf{b})}(\overline{\zeta_1},\zeta_2|\zeta_1,\overline{\zeta_2})
\\
&\times
\overset{k}{\underset{i,j=3}{\det}} 
\bigl(K_{1,2}^{(\mathbf{b})}(\zeta_i,\overline{\zeta_i},\zeta_j,\overline{\zeta_j}|\zeta_1,\overline{\zeta_2}) \bigr), 
\end{split}
\end{align}
uniformly for $\zeta_k$ in compact subsets of $\C$. 
\item[(II)] \textbf{\textup{(Strongly non-unitary regime: edge case)}} 
For $x\in\C$, we define
\begin{equation}
\label{EdgeF1}
F(x)=\frac{1}{2}\erfc\Bigl(\frac{x}{\sqrt{2}}\Bigr),
\end{equation}
and 
\begin{equation}
\label{EdgeF2}
\mathcal{F}(x)=e^{-\frac{1}{2}x^2}-\sqrt{2\pi}xF(x).
\end{equation}
Let 
\begin{equation}
\label{ThmEdge1}
\mathfrak{c}_{\mathfrak{s}}
=
\begin{cases}
\sqrt{\frac{a+b+1}{2\pi(a+1)b}} & \text{if $\mathfrak{s}=1$}, \\
\sqrt{\frac{a+b+1}{2\pi a(b+1)}}  & \text{if $\mathfrak{s}=-1$},
\end{cases}
\end{equation}
and we write 
\begin{align}
\begin{split}
\label{ThmEdge2}
&\quad H(a,b,c,d,f)
\\
&=
-
\frac{\sqrt{2\pi}
\frac{d}{dx}
\Bigl[
e^{\frac{(a+x)^2}{2}}
\Bigl(e^{-f}F(b+x)F(c+x)-F(d+x)F(a+x)+fF(d)F(a+x) \Bigr)
\Bigr]\Bigr|_{x=0}}{e^{\frac{1}{2}a^2}\mathcal{F}(a)}.
\end{split}
\end{align}
Finally, we write 
\begin{equation}
\cK_{1,1}^{(\mathbf{e})}(\overline{\zeta},\eta|\chi,\overline{\chi})
=
\frac{e^{\overline{\zeta}\eta}H(\overline{\chi}+\chi,\overline{\zeta}+\chi,\overline{\chi}+\eta,\overline{\zeta}+\eta,(\overline{\zeta}-\overline{\chi})(\eta-\chi))}{(\overline{\zeta}-\overline{\chi})^2(\eta-\chi)^2}
\end{equation}
and 
\begin{equation}
\omega^{(\mathbf{e})}(\overline{\zeta},\eta|\chi,\overline{\chi})
=
\varpi(\zeta,\overline{\zeta},\eta,\overline{\eta}|\chi,\overline{\chi})e^{-\overline{\zeta}\eta}.
\end{equation}
Then, we have 
\begin{equation}
\label{D11EDGE}
\mathfrak{D}_{1,1}^{(k,\mathbf{e})}(\boldsymbol{\zeta}_{(k)}) 
=
\mathfrak{c}_{\mathfrak{s}}
\mathcal{F}(\zeta_1+\overline{\zeta_1}) 
\overset{k}{\underset{i,j=2}{\det}} 
\bigl(K_{1,1}^{(\mathbf{e})}(\zeta_i,\overline{\zeta_i},\zeta_j,\overline{\zeta_j}|\zeta_1,\overline{\zeta_1}) \bigr), 
\end{equation}
uniformly for $\zeta_k$ in compact subsets of $\C$. 
Moreover, we have 
\begin{align}
\label{D12EDGE}
\begin{split}
\mathfrak{D}_{1,2}^{(k,\mathbf{e})}(\boldsymbol{\zeta}_{(k)})
&=
-
\mathfrak{c}_{\mathfrak{s}}e^{-|\zeta_1-\zeta_2|^2}\mathcal{F}(\zeta_1+\overline{\zeta_2})
\frac{H(\overline{\zeta_2}+\zeta_1,\overline{\zeta_1}+\zeta_1,\overline{\zeta_2}+\zeta_2,\overline{\zeta_2}+\zeta_1,-(\overline{\zeta_1}-\overline{\zeta_2})(\zeta_1-\eta_2))}{(\overline{\zeta_1}-\overline{\zeta_2})^2(\zeta_1-\zeta_2)^2}
\\
&\times
\overset{k}{\underset{i,j=3}{\det}} 
\bigl(K_{1,2}^{(\mathbf{e})}(\zeta_i,\overline{\zeta_i},\zeta_j,\overline{\zeta_j}|\zeta_1,\overline{\zeta_2}) \bigr), 
\end{split}
\end{align}
uniformly for $\zeta_k$ in compact subsets of $\C$. 

\item[(III)] \textbf{\textup{(Weakly non-unitary regime)}} 
For $z\in\C$ and $\rho>0$, we define
\begin{equation}
\label{WeakL1}
L_{\rho}(z)=\frac{1}{\sqrt{2\pi}}\int_{-\frac{\rho}{\sqrt{2}}}^{\frac{\rho}{\sqrt{2}}}e^{-\frac{1}{2}(z-\xi)^2}d\xi,
\end{equation}
and 
\begin{equation}
\label{WeakL2}
\mathcal{L}_{\rho}(z)
=
\frac{1}{\sqrt{2\pi}}
\Bigl\{ \Bigl(z+\frac{\rho}{\sqrt{2}}\Bigr)e^{-\frac{1}{2}(z-\frac{\rho}{\sqrt{2}})^2}
-\Bigl(z-\frac{\rho}{\sqrt{2}}\Bigr)
\Bigl(\sqrt{2\pi}\Bigl(z+\frac{\rho}{\sqrt{2}} \Bigr)L_{\rho}(z)+e^{-\frac{1}{2}(z+\frac{\rho}{\sqrt{2}})^2} \Bigr)
\Bigr\}. 
\end{equation}
We write 
\begin{equation}
\label{ThmWeak1}
\mathcal{A}_{\rho}(a,b,c,d,f)
=\Bigl(\frac{\rho}{\sqrt{2}}+a\Bigr)\Bigl(\frac{\rho}{\sqrt{2}}-a\Bigr)
\Bigl( 
e^{f}(f-1)L_{\rho}(a)L_{\rho}(d)+L_{\rho}(b)L_{\rho}(c)
\Bigr),
\end{equation}
\begin{equation}
\label{ThmWeak2}
\mathcal{B}_{\rho}(a,b,c)
=
L_{\rho}(b)
\Bigl\{
\Bigl(a+\frac{\rho}{\sqrt{2}}\Bigr)\frac{e^{-\frac{1}{2}(c-\frac{\rho}{\sqrt{2}})^2}}{\sqrt{2\pi}}
-
\Bigl(a-\frac{\rho}{\sqrt{2}}\Bigr)\frac{e^{-\frac{1}{2}(c+\frac{\rho}{\sqrt{2}})^2}}{\sqrt{2\pi}}
\Bigr\},
\end{equation}
and 
\begin{equation}
\label{ThmWeak3}
\mathcal{C}_{\rho}(a,b)
=
\frac{e^{-\frac{1}{2}(a-\frac{\rho}{\sqrt{2}})^2-\frac{1}{2}(b+\frac{\rho}{\sqrt{2}})^2}}{\sqrt{2\pi}^2}
+
\frac{e^{-\frac{1}{2}(a+\frac{\rho}{\sqrt{2}})^2-\frac{1}{2}(b-\frac{\rho}{\sqrt{2}})^2}}{\sqrt{2\pi}^2}.
\end{equation}
We also write
\begin{align}
\begin{split}
\label{ThmWeak4}
\mathcal{H}_{\rho}(a,b,c,d,f)
&=
\mathcal{A}_{\rho}(a,b,c,d,f)
+
\mathcal{B}_{\rho}(a,b,c)
+
\mathcal{B}_{\rho}(a,c,b)
+
fe^{f}
\mathcal{B}_{\rho}(a,d,a)
-e^{f}\mathcal{B}_{\rho}(a,d,a)
\\
&
-e^{f}\mathcal{B}_{\rho}(a,a,d)
+
\mathcal{C}_{\rho}(b,c)
-e^{f}
\mathcal{C}_{\rho}(a,d).
\end{split}
\end{align}
Finally, we write 
\begin{equation}
\label{ThmWeak6}
\cK_{1,1}^{(\mathbf{w})}(\overline{\zeta},\eta|\chi,\overline{\chi})=
\frac{\mathcal{H}_{\rho}
(\overline{\chi}+\chi,\overline{\zeta}+\chi,\overline{\chi}+\eta,\overline{\zeta}+\eta,(\overline{\zeta}-\overline{\chi})(\eta-\chi))
}{(\overline{\zeta}-\overline{\chi})^2(\eta-\chi)^2\mathcal{L}_{\rho}(\chi+\overline{\chi})},
\end{equation}
and 
\begin{equation}
\label{ThmWeak7}
\omega^{(\mathbf{w})}(\overline{\zeta},\eta|\chi,\overline{\chi})
=
\varpi(\zeta,\overline{\zeta},\eta,\overline{\eta}|\chi,\overline{\chi})
e^{-(\overline{\zeta}-\overline{\chi})(\eta-\chi)}. 
\end{equation}
Then, we have 
\begin{equation}
\label{D11WEAK}
\mathfrak{D}_{1,1}^{(k,\mathbf{w})}(\boldsymbol{\zeta}_{(k)}) 
=
\mathcal{L}_{\rho}(\zeta_1+\overline{\zeta_1})
\overset{k}{\underset{i,j=2}{\det}} 
\bigl(K_{1,1}^{(\mathbf{w})}(\zeta_i,\overline{\zeta_i},\zeta_j,\overline{\zeta_j}|\zeta_1,\overline{\zeta_1}) \bigr),
\end{equation}
uniformly for $\zeta_k$ in compact subsets of $\C$. 
Moreover, we have 
\begin{align}
\label{D12WEAK}
\mathfrak{D}_{1,2}^{(k,\mathbf{w})}(\boldsymbol{\zeta}_{(k)})
=
-
\mathcal{L}_{\rho}(\zeta_1+\overline{\zeta_2})
\mathcal{K}^{(\mathbf{w})}(\overline{\zeta_1},\zeta_2|\zeta_1,\overline{\zeta_2})
\overset{k}{\underset{i,j=3}{\det}} 
\bigl(K_{1,2}^{(\mathbf{w})}(\zeta_i,\overline{\zeta_i},\zeta_j,\overline{\zeta_j}|\zeta_1,\overline{\zeta_2}) \bigr), 
\end{align}
uniformly for $\zeta_k$ in compact subsets of $\C$. 

\item[(IV)] \textbf{\textup{(Singular origin regime)}} 
We denote the two-parametric Mittag-Lehler function by
\begin{equation}
\label{ThmSingular1}
E_{a,b}(z)=\sum_{k=0}^{\infty}\frac{z^k}{\Gamma(ak+b)}.
\end{equation}
We also write 
\begin{equation}
\label{ThmSingular2}
\mathcal{E}_{1,c}(z|x)=(x-c)E_{1,c+1}(z)+\frac{1}{\Gamma(c)},
\end{equation}
and we define
\begin{align}
\begin{split}
\label{ThmSingular3}
&\quad
\mathcal{S}_L(\overline{\zeta}\chi,\overline{\chi}\eta,\overline{\zeta}\eta,\overline{\chi}\chi,(\overline{\zeta}-\overline{\chi})(\eta-\chi))
\\
&
=
(\chi\overline{\chi}-L)
\Bigl( 
E_{1,L+1}(\overline{\zeta}\chi)E_{1,L+1}(\overline{\chi}\eta)
-(1-(\overline{\zeta}-\overline{\chi})(\eta-\chi))E_{1,L+1}(\overline{\zeta}\eta)E_{1,L+1}(\overline{\chi}\chi)
\Bigr)
\\
&
+\frac{1}{\Gamma(L)}
\Bigl( E_{1,L+1}(\overline{\zeta}\chi)+E_{1,L+1}(\overline{\chi}\eta)-E_{1,L+1}(\overline{\zeta}\eta)-E_{1,L+1}(\overline{\chi}\chi)
+(\overline{\zeta}-\overline{\chi})(\eta-\chi)E_{1,L+1}(\overline{\chi}\chi)
\Bigr).
\end{split}
\end{align}
Finally, we write
\begin{equation}
\label{ThmSingular5}
\cK_{1,1}^{(\mathbf{s})}(\overline{\zeta},\eta|\chi,\overline{\chi})
=
\frac{\mathcal{S}_L(\overline{\zeta}\chi,\overline{\chi}\eta,\overline{\zeta}\eta,\overline{\chi}\chi,(\overline{\zeta}-\overline{\chi})(\eta-\chi))}{(\overline{\zeta}-\overline{\chi})^2(\eta-\chi)^2\mathcal{E}_{1,L}(\overline{\chi}\chi|\overline{\chi}\chi)},
\end{equation}
and 
\begin{equation}
\label{ThmSingular6}
\omega^{(\mathbf{s})}(\overline{\zeta},\eta|\chi,\overline{\chi})
=
\varpi(\zeta,\overline{\zeta},\eta,\overline{\eta}|\chi,\overline{\chi})
(\overline{\zeta}\eta)^{L}e^{-\frac{1}{2}(|\zeta|^2+|\eta|^2)}.
\end{equation}
Then, we have 
\begin{equation}
\label{D11SINGULAR}
\mathfrak{D}_{1,1}^{(k,\mathbf{s})}(\boldsymbol{\zeta}_{(k)}) 
=
|\zeta_1|^{2L-2}\mathcal{E}_{1,L}(|\zeta_1|^2)e^{-|\zeta_1|^2}
\overset{k}{\underset{i,j=2}{\det}} 
\bigl(K_{1,1}^{(\mathbf{s})}(\zeta_i,\overline{\zeta_i},\zeta_j,\overline{\zeta_j}|\zeta_1,\overline{\zeta_1}) \bigr),
\end{equation}
uniformly for $\zeta_k$ in compact subsets of $\C$. 
Moreover, we have 
\begin{align}
\begin{split}
\label{D12SINGULAR}
\mathfrak{D}_{1,2}^{(k,\mathbf{s})}(\boldsymbol{\zeta}_{(k)}) 
&=
-
\frac{\mathcal{E}_{1,L}(\zeta_1\overline{\zeta_2})|\zeta_1|^{2L}|\zeta_2|^{2L}}{\zeta_1\overline{\zeta_2}}e^{-|\zeta_1|^2-|\zeta_2|^2}
\cK_{1,1}^{(\mathbf{s})}(\overline{\zeta_1},\zeta_2|\zeta_1,\overline{\zeta_2})
\\
&\times
\overset{k}{\underset{i,j=3}{\det}} 
\bigl(K_{1,2}^{(\mathbf{s})}(\zeta_i,\overline{\zeta_i},\zeta_j,\overline{\zeta_j}|\zeta_1,\overline{\zeta_2}) \bigr), 
\end{split}
\end{align}
uniformly for $\zeta_k$ in compact subsets of $\C$. 
\end{itemize}
\end{thm}
\begin{rem}
By \eqref{WeakL2}, note that 
\[
\cL_{\rho}(x)\sim \frac{2}{\rho^2},\quad \text{as $\rho\to\infty$}. 
\]
From the definition of \eqref{ThmWeak4}, we have 
\[
\frac{2}{\rho^2}\mathfrak{D}_{1,1}^{(k,\mathbf{w})}(\boldsymbol{\zeta}_{(k)}) \to \mathfrak{D}_{1,1}^{(k,\mathbf{b})}(\boldsymbol{\zeta}_{(k)}),\quad \text{as $\rho\to\infty$}.
\]
This implies that the weak non-unitary regime recovers the strongly non-unitary regime. 
\end{rem}

As we already remarked, we could get the scaling limits of the conditional expectation of the diagonal and the off-diagonal overlaps from \eqref{On conditional expectation 1} and \eqref{Off conditional expectation 1}. 
To state our last result, we now introduce some functions. 
For $z\in\C$, we define 
\[
\Psi_{1,2}^{(\mathbf{b})}(z)=
\frac{1}{|z|^4} \frac{1-(1+|z|^2)e^{-|z|^2}}{1-e^{-|z|^2}}.
\]
For $z,a,b,c,d,f\in\C$, we define 
\begin{align*}
&\Psi_{1,1}^{(\mathbf{e})}(z)
=\frac{\mathcal{F}(z)}{F(z)},
\quad
\Psi_{1,2}^{(\mathbf{e})}(a,b,c,d,f)
=
-\frac{H(a,b,c,d,-|f|^2)e^{-|f|^2}\cF(a)}{|f|^4(F(b)F(c)-e^{-|f|^2}F(a)F(d))},
\\
&\Psi_{1,1}^{(\rho)}(z)
=\frac{\mathcal{L}_{\rho}(z)}{L_{\rho}(z)},
\quad
\Psi_{1,2}^{(\rho)}(a,b,c,d,f)
=
-\frac{\mathcal{H}_{\rho}(a,b,c,d,-|f|^2)}{|f|^4(L_{\rho}(b)L_{\rho}(c)-e^{-|f|^2}L_{\rho}(a)L_{\rho}(d))},
\\
&\Psi_{1,1}^{(L)}(z)
=\frac{\mathcal{E}_{L}(|z|^2)}{|z|^2E_{1,1+L}(|z|^2)},
\quad
\Psi_{1,2}^{(L)}(a,b,c,d,f)
=
-
\frac{\mathcal{S}_L(a,b,c,d,-|f|^2)}{a|f|^2(E_{1,1+L}(b)E_{1,1+L}(c)-E_{1,1+L}(a)E_{1,1+L}(d))}.
\end{align*}

Then, combining Theorem \ref{MainResult} with \eqref{On conditional expectation 1} and \eqref{Off conditional expectation 1}, we can readily find the following: 
\begin{cor}\label{Cor}
\begin{itemize}
\item[(I)] \textbf{\textup{(Strongly non-unitary regime: bulk case)}} 
Conditionally on $\{z_1=p+\zeta_1/\sqrt{N\delta(p)}\}$ for $\mathrm{int}(S)$, we have 
\[
\lim_{N\to\infty}
\frac{1}{N}
\E_N\Bigl[\mathcal{O}_{1,1}\Bigr|z_1=p+\frac{\zeta_1}{\sqrt{N\delta(p)}}\Bigr]=
\frac{b(b+1)}{a+b+1}
\frac{\Bigl(|p|^2-\frac{a}{b+1} \Bigr)\Bigl(\frac{a+1}{b}-|p|^2 \Bigr)}{|p|^2},
\]
uniformly for $\zeta_1$ in a compact subset of $\C$.
Moreover, conditionally on $\{z_1=p+\zeta_1/\sqrt{N\delta(p)},z_2=p+\zeta_2/\sqrt{N\delta(p)}\}$ for $\mathrm{int}(S)$, we have 
\begin{align*}
\lim_{N\to\infty}
&\quad\frac{1}{N}
\E_N\Bigl[\mathcal{O}_{1,2}\Bigr|z_1=p+\frac{\zeta_1}{\sqrt{N\delta(p)}},z_2=p+\frac{\zeta_2}{\sqrt{N\delta(p)}}\Bigr]
\\
&=
-
\frac{b(b+1)}{a+b+1}
\frac{\Bigl(|p|^2-\frac{a}{b+1} \Bigr)\Bigl(\frac{a+1}{b}-|p|^2 \Bigr)}{|p|^2}
\Psi_{1,2}^{(\mathbf{b})}(\zeta_1-\zeta_2),
\end{align*}
uniformly for $\zeta_1,\zeta_2$ in compact subsets of $\C$.
\item[(II)] \textbf{\textup{(Strongly non-unitary regime: edge case)}}
Conditionally on $\{z_1=e^{i\theta}(p+\mathfrak{s}\zeta_1/\sqrt{N\delta(p)})\}$ for $p\in\partial S$ with $\theta\in[0,2\pi)$, we have 
\[
\lim_{N\to\infty}
\frac{1}{\sqrt{N}}
\E_N\Bigl[\mathcal{O}_{1,1}\Bigr|z_1=e^{i\theta}\Bigl( p+\mathfrak{s}\frac{\zeta_1}{\sqrt{N\delta(p)}}\Bigr)\Bigr]=
\mathfrak{c}_{\mathfrak{s}}
\Psi_{1,1}^{(\mathbf{e})}(\zeta_1+\overline{\zeta}_1),
\]
uniformly for $\zeta_1$ in a compact subset of $\C$.
Moreover, conditionally on $\{z_1=p+\mathfrak{s}\zeta_1/\sqrt{N\delta(p)},z_2=p+\mathfrak{s}\zeta_2/\sqrt{N\delta(p)}\}$ for $p\in\partial S$, we have 
\begin{align*}
&\lim_{N\to\infty}
\frac{1}{\sqrt{N}}
\E_N\Bigl[\mathcal{O}_{1,1}\Bigr|z_1=p+\mathfrak{s}\frac{\zeta_1}{\sqrt{N\delta(p)}},z_2=p+\mathfrak{s}\frac{\zeta_2}{\sqrt{N\delta(p)}}\Bigr]
\\
&
=
\mathfrak{c}_{\mathfrak{s}}
\Psi_{1,2}^{(\mathbf{e})}
\left(\zeta_1+\overline{\zeta}_2,\zeta_1+\overline{\zeta}_1,\zeta_2+\overline{\zeta}_2,\zeta_2+\overline{\zeta}_1,\zeta_1-\zeta_2\right),
\end{align*}
uniformly for $\zeta_1,\zeta_2$ in compact subsets of $\C$.
\item[(III)] \textbf{\textup{(Weakly non-unitary regime)}} 
Conditionally on $\{z_1=e^{i\theta}(1+\zeta_1/\sqrt{N\delta(1)})\}$ for $p=1$ with $\theta\in[0,2\pi)$, we have 
\[
\lim_{N\to\infty}
\E_N\Bigl[\mathcal{O}_{1,1}\Bigr|z_1=e^{i\theta}\Bigl( p+\frac{\zeta_1}{\sqrt{N\delta(1)}}\Bigr)\Bigr]
=
\Psi_{1,1}^{(\rho)}(\zeta_1+\overline{\zeta}_1),
\]
uniformly for $\zeta_1$ in a compact subset of $\C$.
Moreover, conditionally on $\{z_1=1+\zeta_1/\sqrt{N\delta(1)},z_2=1+\zeta_2/\sqrt{N\delta(1)}\}$, we have 
\begin{align*}
&\quad\lim_{N\to\infty}
\E_N\Bigl[\mathcal{O}_{1,2}\Bigr|z_1=1+\frac{\zeta_1}{\sqrt{N\delta(1)}},z_2=1+\frac{\zeta_2}{\sqrt{N\delta(1)}}\Bigr]
\\
&=\Psi_{1,2}^{(\rho)}
\left(\zeta_1+\overline{\zeta}_2,\zeta_1+\overline{\zeta}_1,\zeta_2+\overline{\zeta}_2,\zeta_2+\overline{\zeta}_1,\zeta_1-\zeta_2\right)
\end{align*}
uniformly for $\zeta_1,\zeta_2$ in compact subsets of $\C$.
\item[(IV)] \textbf{\textup{(Singular origin regime)}}
Conditionally on $\{z_1=\zeta_1/\sqrt{N\delta(p)}\}$ for $p\in\partial S$, we have 
\[
\lim_{N\to\infty}
\frac{1}{N}
\E_N\Bigl[\mathcal{O}_{1,1}\Bigr|z_1=\frac{\zeta_1}{\sqrt{N\delta(0)}}\Bigr]=
\Psi_{1,1}^{(L)}(\zeta_1),
\]
uniformly for $\zeta_1$ in a compact subset of $\C$.
Moreover, conditionally on $\{z_1=p+\zeta_1/\sqrt{N\delta(0)},z_2=\zeta_2/\sqrt{N\delta(0)}\}$, we have 
\[
\lim_{N\to\infty} 
\frac{1}{N}\E_N\Bigl[\mathcal{O}_{1,2}\Bigr| z_1=\zeta_1,z_2=\zeta_2\Bigr]
=\Psi_{1,2}^{(L)}
\left(\zeta_1\overline{\zeta}_2,\zeta_1\overline{\zeta}_1,\zeta_2\overline{\zeta}_2,\zeta_2\overline{\zeta}_1,\zeta_1-\zeta_2\right),
\]
uniformly for $\zeta_1,\zeta_2$ in a compact subset of $\C$.
\end{itemize}
\end{cor}

\begin{rem}[Universality]
From Theorem~\ref{MainResult} and Corollary~\ref{Cor}, we conclude the universality of the overlaps for the non-Gaussian model. In fact, the results for Ginibre unitary ensemble in \cite{ATTZ} can be recovered in the strongly non-unitary regime up to constants. Also, the results in the weakly non-unitary and the singular origin regimes are same as the results for induced Ginibre unitary ensemble shown in \cite{N23}. 
\end{rem}

\begin{rem}
The differences of the scaling in the strongly non-unitary for the bulk and edge scaling limits and the weakly non-unitary regimes are explained as follows. The bulk scaling limit in the strongly non-unitary regime implies that the conditional expectation of the diagonal overlap is affected by the outer and inner boundaries with $O(N)$. Hence, if we zoom in edge points with $O(N^{-1/2})$, we can detect the corresponding scaling limit. On the other hand, in the weakly non-unitary regime, the inner and outer boundaries are close with order $O(N^{-1})$. Since this scale is canceled out with the bulk scale, an additional scaling is not necessary. 
\end{rem}

\section{Finite $N$-kernel analysis: Proofs of Proposition~\ref{Prop_FiniteN} and Theorem~\ref{Thm_Simplification} } \label{Section4}
In this section, we prove Proposition~\ref{Prop_FiniteN} and Theorem~\ref{Thm_Simplification}. 
Our strategy to construct a family of the planar orthogonal polynomial associated with the weight function \eqref{Weight} follows the moment method as in \cite{ATTZ}. 
For induced Ginibre unitary ensemble, the same strategy was also done in \cite{N23}. 
\subsection{Moment method}
We regard complex variables $z$ and $\overline{z}$ as independent variables. 
Then, there exists a family of bi-orthogonal polynomials associated with the weight function \eqref{Weight} denoted by $\{P_k(\cdot|a,\overline{a}),Q_k(\cdot|a,\overline{a})\}_{k=0}^{\infty}$ on $\C$ such that for $a\in\C$, 
\[
\langle P_j(\cdot|a,\overline{a}),Q_k(\cdot|a,\overline{a})\rangle _{\omega}
= 
\int_{\C}\overline{P_j(z|a,\overline{a})}Q_k(z|a,\overline{a})\omega(z,\overline{z}|a,\overline{a})dA(z)
= 
h_j\delta_{j,k},
\]
where $h_j$ is the norming constant. 
From the elementary linear algebra, we have 
\[
\prod_{2\leq j<k\leq N}|z_j-z_k|^2 \prod_{j=2}^{N}\omega(z_j,\overline{z_j}|z_1,\overline{z_1})
= 
\prod_{j=0}^{N-2}h_j \times \underset{2\leq j,k\leq N}{\det} \Bigl(K_{1,1}^{(N-1)}(z_i,\overline{z_i},z_j,\overline{z_j}| z_1, \overline{z_1} ) \Bigr),
\]
where $K_{1,1}^{(N-1)}$ is an integral kernel defined by 
\[
K_{1,1}^{(N)}(z,\overline{z},w,\overline{w}|a,\overline{a})
= \sum_{k=0}^{N-1}\frac{\overline{P_k(z|a,\overline{a})}Q_k(w|a,\overline{a})}{h_k}
\omega(z,\overline{z}|a,\overline{a}).
\]
Here, it would be convenient to define the reduced polynomial kernel $\cK_{1,1}^{(N)}$ via 
\begin{align*}
K_{1,1}^{(N)}(z,\overline{z},w,\overline{w}|a,\overline{a})&=\cK_{1,1}^{(N)}(\overline{z},w|a,\overline{a})\omega(z,\overline{z}|a,\overline{a}),\\
\cK_{1,1}^{(N)}(\overline{z},w|a,\overline{a})&=\sum_{k=0}^{N-1}\frac{\overline{P_k(z|a,\overline{a})}Q_k(w|a,\overline{a})}{h_k}.
\end{align*}
Then, the $k$-th conditional expectation of the diagonal overlap is given by
\[
D_{1,1}^{(N,k)}(\boldsymbol{z}_{(k)})
=
\frac{N!}{Z_N}\frac{|z_1|^{2L}}{(1+|z_1|^2)^{n+L+1}}\prod_{j=0}^{N-2}h_j\times
\underset{2\leq i,j\leq k}{\det}\Bigl(K_{1,1}^{(N-1)}(z_i,\overline{z_i},z_j,\overline{z_j}|z_1,\overline{z_1})\Bigr).
\]
As we already mentioned, from now on, we focus on the diagonal overlap case. 
The step to get from \eqref{D11ver1} to \eqref{D12ver1} is done similar to \cite[p.13]{ATTZ}. 
Now, we shall prove Proposition~\ref{Prop_FiniteN}. 
\begin{proof}[Proof of Proposition~\ref{Prop_FiniteN}]
We define a moment matrix 
\[
M_{i,j}=\langle z^i,z^j \rangle_{\omega}=\int_{\C}\overline{z}^iz^j\omega(z,\overline{z}|a,\overline{a})dA(z),
\]
where the weight function is defined by \eqref{Weight}.
Since (recall \eqref{Potential})
\[
\int_{\C}|z|^{2k}e^{-NQ(z)}dA(z)=\frac{\Gamma(k+L+1)\Gamma(n-k)}{\Gamma(L+n+1)},
\]
we have 
\begin{equation}
\label{moment1}
M_{i,j}=\frac{\Gamma(n-i-1)\Gamma(i+L+1)}{\Gamma(n+L+1)}\mu_{i,j},
\end{equation}
where
\begin{equation}
\label{moment2}
\mu_{i,j}=(i+L+2+(n-i)|a|^2\delta_{i,j}-(i+L+1)a\delta_{i+1,j}-(n-i-1)\overline{a}\delta_{i,j+1}.
\end{equation}
By the LDU decomposition of $\mu=LDU$, where 
$D_{p,q}=d_p \delta_{p,q}$, $L_{p,q}=\delta_{p,q}+\ell_{p} \delta_{p,q+1}$, and $U_{p,q}=\delta_{p,q}+u_q\delta_{q,p+1}$ for $p,q\in\N\cup\{0\}$, we have 
\[
d_p=-d_{p-1} \ell_{p} u_p \mathbf{1}_{p\geq 1}+p+L+2+|a|^2(n-p),\quad 
u_{p+1}=-\frac{a(p+L+1)}{d_p},\quad 
\ell_{p+1}=-\frac{\overline{a}(n-p-2)}{d_p}.
\]
This implies that 
\[
d_p=-\frac{x(p+L)(n-p-1)}{d_{p-1}}+p+L+2+x(n-p),\quad p\geq 1,
\]
with $d_0=L+2+nx$ for $x=|a|^2$. 
We define a sequence $\{r_p\}_{p=0}^{\infty}$ by $d_p=\frac{r_{p+1}}{r_p}$. 
Then, we have  
\begin{equation}
\label{Rsequence}
r_{p+1}=((n-p)x+p+L+2)r_p-x(p+L)(n-p-1)r_{p-1}, \quad 
r_1=L+2+nx,\quad
r_0=1.
\end{equation}
By the induction argument, the unique solution of \eqref{Rsequence} is given by
\[
r_p=\frac{n!(L+p)!}{(L+n)!}g_{p}^{(n,L)}(x),
\]
where $g_p^{(n,L)}(x)$ is defined by \eqref{GNL1}
After multiplying $\mu$ by $\mathrm{diag}(\Gamma(n-i-1)\Gamma(i+L+1)/\Gamma(n+L+1))_{i=0,1,2,\dots}$, and with the same notation, 
 we find the LDE decomposition of $M=LDU$, 
 \begin{align*}
 L_{p,k}&=\delta_{p,k}-\overline{a}\frac{g_{k-1}^{(n,L)}(x)}{g_k^{(n,L)}(x)}\delta_{p,k+1}, \\
 D_{k,k}&=\frac{\Gamma(k+L+2)\Gamma(n-k-1)}{\Gamma(n+L+1)}\frac{g_{k+1}^{(n,L)}(x)}{g_{k}^{(n,L)}(x)}, \\
 U_{k,p}&=\delta_{k,p}-a\frac{g_{k-1}^{(n,L)}(x)}{g_k^{(n,L)}(x)}\delta_{p,k+1},
 \end{align*}
for $p,k\geq 0$. Here, similar to the same discussion in \cite[p.14]{ATTZ}, we have 
\begin{equation}
\label{LDU result}
h_k=D_{k,k},\quad 
P_k(z|a,\overline{a})=\sum_{m=0}^{k}\overline{L}^{-1}_{k,m}z^m,\quad
Q_k(z|a,\overline{a})=\sum_{m=0}^{k}U^{-1}_{m,k}z^m.
\end{equation}
Then, the reduced polynomial kernel \eqref{FiniteN2} can be written in terms of $L,D,U$:
\[
\cK_{1,1}^{(N)}(z,\overline{z}|a,\overline{a}) = \sum_{i,j=0}^{N-1}\overline{z}^iC_{i,j}^{(N-1)}z^j,\quad 
C_{i,j}^{(N)}=\sum_{k=0}^{N}U^{-1}_{i,k}\frac{1}{D_{k,k}}L^{-1}_{k,j}. 
\]
Note that 
\[
\frac{N!}{Z_N}\prod_{j=0}^{N-2}h_j
= \frac{g_{N-1}^{(n,L)}(x)}{g_0^{(n,L)}(x)}\frac{\Gamma(L+n+1)}{\Gamma(L+1)\Gamma(n)}
= n g_{N-1}^{(n,L)}(x),
\]
where we used $g_0^{(n,L)}(x)=\frac{\Gamma(L+N+1)}{\Gamma(L+1)\Gamma(n+1)}$. 
Notice also that 
\[
L^{-1}_{p,q}=
\begin{cases}
0 & q>p,\\
1 & q=p, \\
\overline{a}^{p-q}\frac{g_q^{(n,L)}(x)}{g_p^{(n,L)}(x)} & q<p,
\end{cases}
\quad
U^{-1}_{p,q}=
\begin{cases}
a^{q-p}\frac{g_p^{(n,L)}(x)}{g_q^{(n,L)}(x)} & q>p, \\
1 & q=p, \\
0 & q<p. 
\end{cases}
\]
Finally, we define 
\[
G_N^{(n,L)}(x|y,z)
= 
\sum_{j,k=0}^{N-1}g_j^{(n,L)}(x)y^jg_k^{(n,L)}(x)z^k\sum_{m=\max(j,s)}^{N-1}
\frac{\Gamma(n+L+1)}{\Gamma(m+L+2)\Gamma(n-m-1)} \frac{x^m}{g_{m+1}^{(n,L)}(x)g_m^{(n,L)}(x)},
\]
which completes the proof.
\end{proof}

\begin{rem}
The planar orthogonal polynomials associated with \eqref{Weight} for $a=0$ are monomials with a norming constant 
\[
h_k(0)=(L+k+2)\frac{\Gamma(n-k-1)\Gamma(L+k+1)}{\Gamma(n+L+1)}.
\]
The corresponding the finite $N$-kernel is given by
\begin{equation}
\cK_{1,1}^{(N)}(z,w|0)=\sum_{k=0}^{N-1}\frac{\Gamma(n+L+1)}{(L+k+2)\Gamma(n-k-1)\Gamma(L+k+1)}(z\overline{w})^k. 
\label{Origin1}
\end{equation}
\end{rem}

\begin{rem}
From the proof of Proposition~\ref{Prop_FiniteN}, we can readily find that 
the planar orthogonal polynomials associated with \eqref{Weight} are explicitly written as 
\begin{equation}
\label{ConditionalP}
P_k(z|a,\overline{a})=\sum_{j=0}^{k}a^{k-j}\frac{g_j^{(n,L)}(a\overline{a})}{g_k^{(n,L)}(a\overline{a})}z^j. 
\end{equation}
Notice also that as in \cite{SMY22}, we can confirm that 
$\{P_k(\cdot|a,\overline{a})\}_k$ satisfy the non-standard three term recurrence relationship
\begin{equation}
\label{ThreeTerm}
zP_k(z|a,\overline{a})
=
P_{k+1}(z|a,\overline{a})+b_kP_k(z|a,\overline{a})+z c_kP_{k-1}(z|a,\overline{a}),
\end{equation}
where 
\begin{equation}
b_k=-a\frac{g_{k}^{(n,L)}(a\overline{a})}{g_{k+1}^{(n,L)}(a\overline{a})},
\quad 
c_k=a\frac{g_{k-1}^{(n,L)}(a\overline{a})}{g_{k}^{(n,L)}(a\overline{a})}.
\end{equation}
Then, the finite $N$-kernel can be also written in terms of $\{P_k(\cdot|a,\overline{a})\}_k$
\begin{align}
\cK_{1,1}^{(N)}(z,w|a,\overline{a})
= & \sum_{k=0}^{N-1}\frac{P_k(z|a,\overline{a})\overline{P_k(w|a,\overline{a})}}{h_k}\\
= & \sum_{k=0}^{N-1}\frac{\Gamma(n+L+1)}{\Gamma(k+L+2)\Gamma(n-k-1)}
\frac{g_{k}^{(n,L)}(a\overline{a})}{g_{k+1}^{(n,L)}(a\overline{a})}
P_k(z|a,\overline{a})\overline{P_k(w|a,\overline{a})}.
\nonumber
\end{align}
Although we do not pursue here, it would be interesting to find a Christoffel-Darboux type identity as in \cite{SMY22} and to analyze the corresponding to differential equation.
In this case, we expect that the corresponding differential equation is the second order differential equation different form \cite{SMY22}. 
Indeed, we shall consider $a=0$. Then, we have 
\[
\cK_{1,1}^{(N)}(z,w|0)=\sum_{k=0}^{N-1}\frac{\Gamma(n+L+1)}{(L+k+2)\Gamma(n-k-1)\Gamma(L+k+1)}(z\overline{w})^k.
\]
We write 
\[
\widehat{\mathcal{K}}_{1,1}^{(N)}(z,w|0)=\frac{(z\overline{w})^{L+2}\cK_{1,1}^{(N)}(z,w|0)}{(1+z\overline{w})^{n+L-1}}.
\]
Then, the above satisfies the following differential equation:
\begin{align*}
& \Bigl\{ z(1+z\overline{w})\partial_z^2 +\bigl( z\overline{w}(n+L-1)-1\bigr) \partial_z- (n+L-1)\overline{w} \Bigr\}
\widehat{\cK}_{1,1}^{(N)}(z,w|0) 
\\
= & 
 \frac{(L+n) \overline{w}\Gamma(n+L+1)(z\overline{w})^{L+1}}{(1+z\overline{w})^{n+L}\Gamma(n)\Gamma(L)}
- \frac{(L+n) \overline{w}\Gamma(n+L+1)(z\overline{w})^{N+L+1}}{(1+z\overline{w})^{n+L}\Gamma(n-N)\Gamma(N+L)}. 
\end{align*}
This suggests that for the general case, we would have a second order differential equation depending on the parameter $a\in\C$.
\end{rem}

\subsection{Simplification step for the finite $N$-kernel}
In this subsection, we will find the simplified representation of \eqref{FiniteN1}, that is, we will prove Theorem~\ref{Thm_Simplification}. 
First, we can rewrite \eqref{FiniteN1} as 
\begin{equation}
\label{GN1}
G_N^{(n,L)}(x|y,z)=\sum_{s,t=0}^{N-1}
g_s^{(n,L)}(x) y^s g_t^{(n,L)}(x) z^t \Bigl( \Phi_{N-1}^{(n,L)}(x)-\Phi_{\max\{s,t\}-1}^{(n,L)}(x)\Bigr),
\end{equation}
where 
\[
\Phi_{\ell}^{(n,L)}(x)=\sum_{j=0}^{\ell}\frac{\Gamma(n+L+1)}{\Gamma(j+L+2)\Gamma(n-j-1)}
\frac{x^j}{g_{j+1}^{(n,L)}(x)g_{j}^{(n,L)}(x)}. 
\]
Let us denote 
\[
\Phi_{\ell,m}^{(n,L)}(x)=\sum_{j=m}^{\ell}\frac{\Gamma(n+L+1)}{\Gamma(j+L+2)\Gamma(n-j-1)} 
\frac{x^j}{g_{j+1}^{(n,L)}(x)g_{j}^{(n,L)}(x)}. 
\]
\begin{lem}\label{Lem_DPhi}
We have 
\[
\Phi_{q}^{(n,L)}(x)= \frac{\Gamma(n+1)\Gamma(L+1)}{\Gamma(n+L+1)}
\frac{(n-1)x-(L+1)}{x(x-L/n)}
+
\frac{-(n-q-2)x+L+q+2}{x(x-L/n)g_{q+1}^{(n,L)}(x)}. 
\]
\end{lem}
\begin{proof}
By the induction argument, we can show that 
\begin{align*}
\Phi_{m+q,m}^{(n,L)}(x)&=\frac{\Gamma(n+L+1)}{\Gamma(L+q+m+2)\Gamma(n-m-1)}
\frac{x^m}{g_{m+q+1}^{(n,L)}(x)g_{m}^{(n,L)}(x)}
\\
&\times
\left\{
\sum_{k=0}^q\frac{(q+1-k)\Gamma(n-m)\Gamma(L+q+m+2)}{\Gamma(n-m-k)\Gamma(L+k+m+2)}x^k\right.\nonumber\\
&\left.
-(L+m+1)\sum_{k=0}^q\frac{(q-k)\Gamma(n-m-1)\Gamma(L+q+m+2)}{\Gamma(n-m-k-1)\Gamma(L+k+m+3)}x^k
\right\}.
\end{align*}
By taking $m=0$ and rearranging the summations, we have
\begin{align*}
\Phi_{q}^{(n,L)}(x)&=
\frac{\Gamma(n+1)\Gamma(L+1)}{\Gamma(n+L+1)g_{q+1}^{(n,L)}(x)}
\\
&\times
\Bigl\{
\frac{(n-1)x-(L+1)}{x}\sum_{k=0}^{q}\frac{(q+1-k)\Gamma(n+L+1)}{\Gamma(n-k)\Gamma(L+k+2)}x^k
+\frac{q+1}{x}\frac{\Gamma(n+L+1)}{\Gamma(n)\Gamma(L+1)}
\Bigr\}
\\
&=
\frac{\Gamma(n+1)\Gamma(L+1)}{\Gamma(n+L+1)g_{q+1}^{(n,L)}(x)}
\\
&\times
\Bigl\{
\frac{(n-1)x-(L+1)}{x}\frac{g_{q+1}^{(n,L)}(x)}{x-\frac{L}{n}}+\frac{\Gamma(n+L+1)}{\Gamma(n+1)\Gamma(L+1)}\frac{-(n-q-2)x+L+q+2}{x\left(x-\frac{L}{n}\right)}
\Bigr\},
\end{align*}
which completes the proof. 
\end{proof}
We are ready to show Theorem~\ref{Thm_Simplification}. 
\begin{proof}[Proof of Theorem~\ref{Thm_Simplification}]
We write 
\begin{equation}
\label{AlphaSPH}
\alpha_m^{(n,L)}(x,\omega)=\sum_{s=0}^{m}g_s^{(n,L)}(x)\omega^s.
\end{equation}
Then, using \eqref{GNL2}, we find that \eqref{AlphaSPH} can be expressed as 
\begin{align}
\label{C1}
\alpha_{m}^{(n,L)}(x,\omega)
&=
\frac{x-\frac{L}{n}}{x}\frac{1}{(1-\omega)^2}
\widehat{q}_{m}^{(n,L)}(x\omega|x)
-\frac{x-\frac{L}{n}}{x}\frac{(x\omega)^{m+1}\Gamma(L+n+1)}{(1-\omega)(1+x)\Gamma(L+m+1)\Gamma(n-m)}
\\
-&\frac{x-\frac{L}{n}}{x}\frac{\omega^{m+1}(L+m+1-(n-m-1)x)\widehat{q}_{m}^{(n,L)}(x|x)}{(1+x)(1-\omega)}
-\frac{x-\frac{L}{n}}{x}\frac{\omega^{m+1}\widehat{q}_{m}^{(n,L)}(x|x)}{(1-\omega)^2}. 
\nonumber
\end{align}
By differentiating \eqref{C1} with respect to $\omega$, we have 
\begin{align}
\begin{split}
\label{C2}
&\omega\partial_\omega\alpha_m^{(n,L)}(x,\omega)\\
&=
\left(m+1+\frac{2\omega}{1-\omega}\right)\alpha_m^{(n,L)}(x,\omega)
+\frac{x-\frac{L}{n}}{x}\frac{\omega^{m+2}(L+m+1-(n-m-1)x)\widehat{q}_{m}^{(n,L)}(x|x)}{(1-\omega)^2(1+x)}
\\
&
-\frac{x-\frac{L}{n}}{x}\frac{(m+1)\widehat{q}_{m}^{(n,L)}(x\omega|x)}{(1-\omega)^2}
+\frac{x-\frac{L}{n}}{x}\frac{(nx\omega-L)\widehat{q}_{m}^{(n,L)}(x\omega|x\omega)}{(1+x\omega)(1-\omega)^2}
\\
&
-\frac{x-\frac{L}{n}}{x}\frac{(x\omega)^{m+1}(1+x+x\omega)}{(1-\omega)(1+x)(1+x\omega)}\frac{\Gamma(L+n+1)}{\Gamma(L+m+1)\Gamma(n-m)}\\
&=\left(m+1+\frac{2\omega}{1-\omega}\right)\alpha_m^{(n,L)}(x,\omega)+R_{m}(x,\omega),
\end{split}
\end{align}
where 
\begin{align}
\begin{split}
\label{C3}
R_{m}(x,\omega)
&=
\frac{x-\frac{L}{n}}{x}\frac{\omega^{m+2}(L+m+1-(n-m-1)x)\widehat{q}_{m}^{(n,L)}(x|x)}{(1-\omega)^2(1+x)}
\\
&
+\frac{x-\frac{L}{n}}{x}\frac{(nx\omega-L)\widehat{q}_{m}^{(n,L)}(x\omega|x\omega)}{(1+x\omega)(1-\omega)^2}
-\frac{x-\frac{L}{n}}{x}\frac{m+1}{(1-\omega)^2}\widehat{q}_{m}^{(n,L)}(x\omega|x)
\\
&
-\frac{x-\frac{L}{n}}{x}\frac{(x\omega)^{m+1}(1+x+x\omega)}{(1-\omega)(1+x)(1+x\omega)}\frac{\Gamma(L+n+1)}{\Gamma(L+m+1)\Gamma(n-m)}.
\end{split}
\end{align}
To compute the summation for the index $\max\{s,t\}-1$, we need 
\begin{equation}
\label{C4}
\sum_{0\leq s<t\leq N-1}g^{(n,L)}_s(x)y^sz^t=\frac{z}{1-z}\alpha_{N-1}^{(n,L)}(x,yz)-\frac{z^N}{1-z}\alpha_{N-1}^{(n,L)}(x,y),
\end{equation}
\begin{align}
\begin{split}
\label{C5}
\sum_{0\leq s<t\leq N-1}g^{(n,L)}_s(x)ty^sz^t
&=\frac{z}{(1-z)^2}\alpha_{N-1}^{(n,L)}(x,yz)
+\frac{z^2}{1-z}\partial_z\alpha_{N-1}^{(n,L)}(x,yz)
\\
&
+\frac{(N-1)z^{N+1}-Nz^N}{(1-z)^2}\alpha_{N-1}^{(n,L)}(x,y).
\end{split}
\end{align}
These follow from straightforward calculations. 
Together with \eqref{C1}, \eqref{C2}, \eqref{C3}, \eqref{C4}, and \eqref{C5}, we can express \eqref{GN1} as 
\begin{align*}
G_{N}^{(n,L)}(x|y,z)
&=
\frac{-(n-N-1)x+L+N+1}{x\left(x-\frac{L}{n}\right)g_{N}^{(n,L)}(x)}\alpha_{N-1}^{(n,L)}(x,y)\alpha_{N-1}^{(n,L)}(x,z)
\\
&
-\frac{x+1}{x\left(x-\frac{L}{n}\right)}\frac{1-yz}{(1-y)(1-z)}R_{N-1}(x,y,z)
\\
&
-\frac{(1-yz)\Big\{(1+x)(1-yz)+(N+L-(n-N)x)(1-y)(1-z)\Bigr\}\alpha_{N-1}^{(n,L)}(x,yz)}{x\left(x-\frac{L}{n}\right)(1-y)^2(1-z)^2}
\\
&
+\frac{z^N((N+L-(n-N)x)(1-z)+1+x)\alpha_{N-1}^{(n,L)}(x,y)}{(1-z)^2x\left(x-\frac{L}{n}\right)}
\\
&
+\frac{y^N((N+L-(n-N)x)(1-y)+1+x)\alpha_{N-1}^{(n,L)}(x,z)}{(1-y)^2x\left(x-\frac{L}{n}\right)}. 
\end{align*}
After long and involved, but simple, computations we have 
\begin{equation}
\label{C7}
x^2G_{N}^{(n,L)}(x|y,z)
=
\frac{T_A^{(n,L)}(x,y,z)}{(1-y)^2(1-z)^2}+\frac{T_B^{(n,L)}(x,y,z)}{(1-y)(1-z)}
+\frac{T_C^{(n,L)}(x,y,z)}{(1-y)^2(1-z)}+\frac{T_C^{(n,L)}(x,z,y)}{(1-y)(1-z)^2},
\end{equation}
where we set 
\begin{align}
\begin{split}
\label{C8}
&\quad \widehat{g}_N^{(n,L)}(x)T_A^{(n,L)}(x,y,z)
\\
&=(1+x)\left(N+L+1+x\right)\Bigl\{
\widehat{q}_{N}^{(n,L)}(xy|x)
\widehat{q}_{N}^{(n,L)}(xz|x)
-\widehat{q}_{N}^{(n,L)}(xyz|x)\widehat{q}_{N}^{(n,L)}(x|x)
\Bigr\}
\\
&-(1+x)(n-N)x
\Bigl\{\widehat{q}_{N-1}^{(n,L)}(xy|x)
\widehat{q}_{N-1}^{(n,L)}(xz|x)
-\widehat{q}_{N-1}^{(n,L)}(xyz|x)\widehat{q}_{N-1}^{(n,L)}(x|x)
\Bigr\},
\end{split}
\end{align}
\begin{align}
\begin{split}
\label{C9}
&\quad\widehat{g}_N^{(n,L)}(x)T_B^{(n,L)}(x,y,z)
\\
&=
\frac{(n-N)\Gamma(L+n+1)^2}{\Gamma(L+N+1)\Gamma(L+N)\Gamma(n+1-N)^2}x(xy)^N(xz)^N
\\
&+
\frac{(n-N)(L+N-(n-N)x)\Gamma(L+n+1)}{\Gamma(L+N+1)\Gamma(n+1-N)}x(xyz)^N\widehat{q}_{N-1}^{(n,L)}(x|x)
\\
&+
\Bigl\{
\frac{(L+n)xq_{N-1}^{(n,L)}(xyz)}{1+xyz}-\frac{x\Gamma(L+n+1)}{(1+xyz)\Gamma(n+1)\Gamma(L)}
+\frac{x(xyz)^N\Gamma(L+n+1)}{(1+xyz)\Gamma(L+N)\Gamma(n+1-N)}
\Bigr\}
\\
&\times
\Bigl\{
(N+L+1+x)\widehat{q}_{N}^{(n,L)}(x|x)-(n-N)x\widehat{q}_{N-1}^{(n,L)}(x|x)
\Bigr\},
\end{split}
\end{align}
and
\begin{align}
\begin{split}
\label{C10}
&\quad\widehat{g}_N^{(n,L)}(x)T_C^{(n,L)}(x,y,z)
\\
&=
\frac{(1+x)x(xyz)^N\Gamma(L+n+1)}{\Gamma(L+N+1)\Gamma(n-N)}
\widehat{q}_{N}^{(n,L)}(x|x)
-\frac{(1+x)x(xz)^N\Gamma(L+n+1)}{\Gamma(L+N+1)\Gamma(n-N)}
\widehat{q}_{N}^{(n,L)}(xy|x). 
\end{split}
\end{align}
We write 
\begin{equation}
\label{C12}
W_{N}^{(n,L)}(x,y,z)
=
\widehat{q}_{N}^{(n,L)}(xy|x)\widehat{q}_{N}^{(n,L)}(xz|x)
-\Bigl( 1-\frac{x(1-y)(1-z)}{1+x}\Bigr)\widehat{q}_{N}^{(n,L)}(x|x)\widehat{q}_{N}^{(n,L)}(xyz|x). 
\end{equation}
Then, we have 
\begin{align*}
W_N^{(n,L)}(x,y,z)&=
W_{N+1}^{(n,L)}(x,y,z)
-\frac{\Gamma(L+n+1)(xy)^{N+1}\widehat{q}_{N+1}^{(n,L)}(xz|x)}{\Gamma(L+N+2)\Gamma(n-N)}
-\frac{\Gamma(L+n+1)(xz)^{N+1}\widehat{q}_{N+1}^{(n,L)}(xy|x)}{\Gamma(L+N+2)\Gamma(n-N)}
\\
&+
\left(1-\frac{x(1-y)(1-z)}{1+x}\right)
\frac{\Gamma(L+n+1)}{\Gamma(L+N+2)\Gamma(n-N)}(xyz)^{N+1}
\widehat{q}_{N+1}^{(n,L)}(x|x)
\\
&+
\left(1-\frac{x(1-y)(1-z)}{1+x} \right)
\frac{\Gamma(L+n+1)}{\Gamma(L+N+2)\Gamma(n-N)}x^{N+1}
\widehat{q}_{N+1}^{(n,L)}(xyz|x)
\\
&+
\frac{x(1-y)(1-z)}{1+x}\frac{\Gamma(L+n+1)^2(xy)^{N+1}(xz)^{N+1}}{\Gamma(L+N+2)^2\Gamma(n-N)^2}.
\end{align*}
Together with the above identity, \eqref{C7} with \eqref{C8}, \eqref{C9}, \eqref{C10}, and \eqref{C12}, 
\begin{align*}
&\quad x^2 \widehat{g}_N^{(n,L)}(x) G_N^{(n,L)}(x,y,z)
\\
&= 
\frac{(1+x)}{(1-y)^2(1-z)^2}
\Bigl((L+N+1)W_{N+1}^{(n,L)}(x,y,z)-x(n-N-1)W_N^{(n,L)}(x,y,z) \Bigr)
+
\frac{R_N^{(n,L)}(x,y,z)}{(1-y)(1-z)}, 
\end{align*}
where
\begin{align*}
&\quad R_N^{(n,L)}(x,y,z) 
\\
&=
\widehat{g}_N^{(n,L)}(x)r_N^{(n,L)}(x,y,z)
\\
&
-\frac{x^2 (n-N-1) \Gamma (L+n+1) (x y z)^N \widehat{q}_{N-1}^{(n,L)}(x|x)}{\Gamma (L+N+1) \Gamma (n-N)}+\frac{x^2 (n-N-1) \Gamma (L+n+1) (x y z)^N \widehat{q}_{N}^{(n,L)}(x|x)}{\Gamma (L+N+1) \Gamma (n-N+1)}
\\
&
-\frac{x \Gamma (L+n+1) (x y z)^{N+1} \widehat{q}_{N+1}^{(n,L)}(x|x)}{\Gamma (L+N+1) \Gamma (n-N)}+\frac{x (L+N+1) (n-N-1) \Gamma (L+n+1) (x y z)^N \widehat{q}_{N}^{(n,L)}(x|x)}{\Gamma (L+N+1) \Gamma (n-N+1)}, 
\end{align*}
and
\begin{align*}
r_N^{(n,L)}(x,y,z)
&=
x(n-N-1)\widehat{q}_{N-1}^{(n,L)}(xyz|x)
\\
&
+
\frac{x}{(1+xyz)}
\Bigl\{
(L+N)\widehat{q}_{N}^{(n,L)}(xyz|xyz)-xyz(n-N)\widehat{q}_{N-1}^{(n,L)}(xyz|xyz)
\\
&
-
\Bigl(
\frac{N(1+xyz)}{nxyz-L}+ \frac{(n-N)(1+xyz)}{nx-L}
\Bigr)
\frac{\Gamma(n+L+1)}{\Gamma(n+1)\Gamma(L)} 
\Bigr\}. 
\end{align*}
By taking $x\mapsto \lambda\overline{\lambda},y\mapsto\overline{z}/\overline{\lambda},z\mapsto w/\lambda$ and some computations, we complete the proof. 
\end{proof}

\section{Proof of main results}\label{Section5}
In this section, we conclude the proof of the main result in this note.
As already seen in \cite{ATTZ}, once we know the scaling limit of the joint averaged $k$-th conditional on-diagonal overlap, 
we can readily find the joint averaged $k$-th conditional off-diagonal overlap via Lemma~\ref{Decoupling}.
Our proof is strongly inspired by \cite{SFv3}. 
We first collect the asymptotic behavior of \eqref{HQNL}. 
\begin{lem}\label{Lem_Asym}
Let 
\[
\delta_N(p)=\frac{n+L+1}{N}\frac{1}{(1+|p|^2)^2},\quad\text{$p\in\mathrm{clo}(S)$}. 
\]
\begin{description}
\item[(I) Strong non-unitarity regime]
Let $L=aN$ and $n=(b+1)N$ with fixed $a\geq 0$ and $b>0$.
For $k=0,1,2$, we have 
\begin{align}
\begin{split}
&\quad\widehat{q}_{N+k-1}^{(n,L)}(\overline{z}w|\lambda\overline{\lambda})
\widehat{\omega}^{(n,L)}(\overline{z},w)\\
&=
\begin{cases}
\frac{e^{-\frac{1}{2}(|\zeta|^2+|\eta|^2)+\overline{\zeta}\eta}}{1+|p|^2}\left(1+O(e^{-\epsilon N})\right)&
 \text{(bulk)},\\ 
\frac{e^{-\frac{1}{2}(|\zeta|^2+|\eta|^2)+\overline{\zeta}\eta}}{1+p^2}
\left(F(\overline{\zeta}+\eta)+\frac{C_{\mathrm{out}}(\zeta,\eta)}{\sqrt{N}}+\sqrt{\frac{a+b+1}{2\pi(a+1)bN}}ke^{-\frac{1}{2}(\overline{\zeta}+\eta)^2}
+O\left(\frac{1}{N}\right))
\right)
 & \text{(outer)},\\
\frac{e^{-\frac{1}{2}(|\zeta|^2+|\eta|^2)+\overline{\zeta}\eta}}{1+p^2}
\left(F(\overline{\zeta}+\eta)-\frac{1}{\chi+\overline{\chi}}\frac{e^{-\frac{1}{2}(\overline{\zeta}+\eta)^2}}{\sqrt{2\pi}}
+\frac{C_{\mathrm{out}}(\zeta,\eta)}{\sqrt{N}}+O\left(\frac{1}{N}\right))
\right)& \text{(inner)},
\end{cases}
\end{split}
\end{align}
as $N\to\infty$, uniformly for $\zeta,\eta,\chi$ in compact subsets of $\C$. 
Here, $F$ is defined by \eqref{EdgeF1}. 
\item[(II) Weakly non-unitarity regime]
Let $L=\frac{N^2}{\rho^2}-N$ and $n=\frac{N^2}{\rho^2}$ with fixed $\rho> 0$. For $z=e^{i\theta}\left(1+\frac{\zeta}{\sqrt{N\delta_N(1)}}\right),w=e^{i\theta}\left(1+\frac{\eta}{\sqrt{N\delta_N(1)}}\right)$ with $\theta\in[0,2\pi)$, we have that for $k=0,1,2$,
\begin{align}
\begin{split}
&\quad
\widehat{q}_{N-1+k}^{(n,L)}(\overline{z}w|\lambda\overline{\lambda})
\widehat{\omega}^{(n,L)}(\overline{z},w)
\\
&=
\frac{e^{-\frac{1}{2}(|\zeta|^2+|\eta|^2)+\overline{\zeta}\eta}}{2}
\Bigl( 
\mathcal{J}_{\rho}(\overline{\zeta},\eta|\overline{\chi},\chi)
+
k
\frac{\rho}{N}
\frac{e^{-\frac{1}{2}\bigl(\overline{\zeta}+\eta-\frac{\rho}{\sqrt{2}}\bigr)^2}}{\sqrt{\pi}}
+
\frac{C_{\mathrm{w}}(\zeta,\eta)}{N}+ O(N^{-2})
\Bigr),
\end{split}
\end{align}
as $N\to\infty$,
where some constant $C_{\mathrm{w}}(\zeta,\eta)$ depends on $\zeta,\eta,\rho$, and
\begin{equation}
\label{LemJ}
\mathcal{J}_{\rho}(\overline{\zeta},\eta|\overline{\chi},\chi)=
L_{\rho}(\overline{\zeta}+\eta)+\frac{e^{-\frac{1}{2}\bigl(\overline{\zeta}+\eta+\frac{\rho}{\sqrt{2}} \bigr)^2}}{\sqrt{2\pi}(\chi+\overline{\chi}+\frac{\rho}{\sqrt{2}})}.
\end{equation}
Here, the convergence is uniform for $\zeta,\eta,\chi$ in compact subsets of $\C$, and $L_{\rho}(z)$ is defined by \eqref{WeakL1}.
\item[(III) At singular points regime]
Let $L\geq 0$ be fixed and $n=(b+1)N$ with fixed $b>0$. For $z=\frac{\zeta}{\sqrt{N\delta_N(0)}},w=\frac{\eta}{\sqrt{N\delta_N(0)}}$ with $\zeta,\eta$ in a compact subset of $\C$, we have that for $k=0,1,2$,
\begin{equation}
\widehat{q}_{N+k-1}^{(n,L)}(\overline{z}w|\lambda\overline{\lambda})
\widehat{\omega}_N^{(n,L)}(\overline{z},w)
=
\frac{1}{\chi\overline{\chi}-L}(\overline{\zeta}\eta)^L\mathcal{E}_{1,L}\left(\overline{\zeta}\eta|\overline{\chi}\chi\right)e^{-\frac{1}{2}(|\zeta|^2+|\eta|^2)}\left(1+o(1)\right),
\end{equation}
uniformly for $\zeta,\eta,\chi$ in a compact subset of $\C$ and we omitted the cocycle factor. 
Here, $E_{a,b}(z)$ is defined by \eqref{ThmSingular1}, and $\mathcal{E}_{1,c}(z|x)$ is defined by \eqref{ThmSingular2}. 
\end{description}
\end{lem}
\begin{rem}
In Lemma~\ref{Lem_Asym}, we have omitted a co-cycle factor for the simplicity of the notation. 
Indeed, it does not affect the results in this note since we can readily find that a co-cycle factor is canceled out or put together in a leading term from the form of \eqref{PartHN}, \eqref{PartQN}, and \eqref{PartFN}. 
\end{rem}
\begin{proof}
Our proof is based on \cite[proof of Lemma 3.5.]{SFv3}. 
Therefore, we firstly assume that $z,w\in\R$, and we also assume that the parameters $L,n$ are integers. 
Then, following the same strategy by \cite[proof of Lemma 3.5.]{SFv3}, 
we prove Lemma~\ref{Lem_Asym} by the probabilistic argument. 
And after that, we extend the validity to the complex variables $z,w\in\C$ by Vitali's Lemma. 
For the detailed discussion of the extension to complex variables $z,w\in\C$, we refer to \cite[proof of Lemma 3.5.]{SFv3}. 
Also, for the notational convenience, we still use the complex conjugate notation. 
Moreover, we may omit a a co-cycle factor in each resulting asymptotic behavior to lighten the notation. 
Now, we shall prove the claim in (I). 
Through proof, we also assume that $p\in\R$ from the perspective of Vitali's argument. 
Let 
\[
\mathfrak{p}=\frac{z\overline{w}}{1+z\overline{w}},\quad \text{for $z,w\in\R$}.
\] 
Note that 
\begin{align*}
\begin{split}
q_{N-1}^{(n,L)}(z \overline{w})
\widehat{\omega}_N^{(n,L)}(z,\overline{w})
&=
\mathfrak{p}^{-L}(1-\mathfrak{p})^{-n}\sum_{k=0}^{N-1}\binom{n+L}{k+L}\mathfrak{p}^{k+L}(1-\mathfrak{p})^{n-k}e^{-\frac{N}{2}(Q(z)+Q(w))}
\\
&=
\mathfrak{p}^{-L}(1-\mathfrak{p})^{-n}\mathfrak{q}_{N-1}^{(n,L)}(\mathfrak{p})\widehat{\omega}_N^{(n,L)}(z,\overline{w}).
\end{split}
\end{align*}
Let $X\overset{d}{\sim}B(n+L,\mathfrak{p})$ be the binomial distribution. Then, we can express 
\begin{align*}
\mathfrak{q}_{N-1}^{(n,L)}(\mathfrak{p})
&=
\P\left(L\leq X\leq N+L-1\right)\\
&=
\P\left(\frac{L-(n+L)\mathfrak{p}}{\sqrt{(n+L)\mathfrak{p}(1-\mathfrak{p})}}\leq \frac{X-(n+L)\mathfrak{p}}{\sqrt{(n+L)\mathfrak{p}(1-\mathfrak{p})}}\leq \frac{N+L-1-(n+L)\mathfrak{p}}{\sqrt{(n+L)\mathfrak{p}(1-\mathfrak{p})}}\right).
\end{align*} 
\par
\textbf{(I) Strongly non-unitary regime}:
under the setting of (I), by Taylor expansion, we have 
\begin{align*}
\frac{N+L-1-(n+L)\mathfrak{p}}{\sqrt{(n+L)\mathfrak{p}(1-\mathfrak{p})}}
&=
\frac{a+1-bp^2}{\sqrt{a+b+1}p}\sqrt{N}-\frac{(1+p^2)(a+1+bp^2)}{2(a+b+1)p^2}(\zeta+\overline{\eta})+O\left(\frac{1}{\sqrt{N}}\right),
\\
\frac{L-(n+L)\mathfrak{p}}{\sqrt{(n+L)\mathfrak{p}(1-\mathfrak{p})}}
&=
\frac{a-(b+1)p^2}{\sqrt{a+b+1}p}\sqrt{N}-\frac{(1+p^2)(a+(1+b)p^2)}{2(a+b+1)p^2}(\zeta+\overline{\eta})+O\left(\frac{1}{\sqrt{N}}\right).
\end{align*}
Then, for $p\in\mathrm{int}S_{(\mathbf{r})}$, by the Gaussian approximation of the binomial distribution, we have
\[
\P\left(L\leq X\leq N+L-1\right)=\P\left(-\infty\leq \mathcal{N}\leq \infty\right)+O\left(e^{-cN}\right)=1+O\left(e^{-cN}\right),
\]
where $\mathcal{N}$ denotes the standard Gaussian distribution and $c$ is a positive constant.  
On the other hand, for the outer edge case, we have
\[
\P\left(L\leq X\leq N+L-1\right)
=\P\left(-\infty\leq \mathcal{N}\leq -(\zeta+\overline{\eta})\right)
+\frac{C_{\mathrm{out}}(\zeta,\overline{\eta})}{\sqrt{N}}
+O\left(\frac{1}{N}\right). \label{eq:outerv}
\]
Similarly, for the inner edge case, we have
\[
\P\left(L\leq X\leq N+L-1\right)
=\P\left(-\infty\leq \mathcal{N}\leq -(\zeta+\overline{\eta})\right)
+\frac{C_{\mathrm{in}}(\zeta,\overline{\eta})}{\sqrt{N}}
+O\left(\frac{1}{N}\right),
\]
where
\[
\P\left(-\infty\leq \mathcal{N}\leq -(\zeta+\overline{\eta})\right)=\frac{1}{2}\mathrm{erfc}\left(\frac{\zeta+\overline{\eta}}{\sqrt{2}}\right)=F(\zeta+\overline{\eta}).
\]
Here, $C_{\mathrm{out}}(\zeta,\overline{\eta})$ and $C_{\mathrm{in}}(\zeta,\overline{\eta})$ only depend on $\zeta,\overline{\eta},a,b$.
By Taylor expansion, we have 
\begin{align*}
\begin{split}
\widehat{\omega}^{(n,L)}(z,\overline{w})
&=
\left(1+O(N^{-1/2})\right)
\frac{p^{2aN}}{(1+p^2)^{(a+b+1)N+1}}
\exp\left(\pm\frac{a-(b+1)p^2}{2p\sqrt{a+b+1}}(\zeta+\overline{\zeta}+\eta+\overline{\eta})\sqrt{N}\right)
\\
&\quad\times
\exp\left(\frac{a-(b+1)p^2}{2(a+b+1)}\frac{1+p^2}{p^2}(|\zeta|^2+|\eta|^2)
\right)
\\
&\quad\times
\exp\left(-\frac{a(1+p^2)+(a-(b+1)p^2)p^2}{4(a+b+1)p^2}\left((\zeta+\overline{\zeta})^2+(\eta+\overline{\eta})^2\right)
\right),
\end{split}
\end{align*}
and 
\begin{align*}
\begin{split}
&\quad \mathfrak{p}^{-L}(1-\mathfrak{p})^{-n}
\\
&=
\left(1+O(N^{-1/2})\right)
p^{-2aN}(1+p^2)^{(a+b+1)N}
\exp\left(\pm\frac{(-a+(b+1)p^2)}{p\sqrt{a+b+1}}(\zeta+\overline{\eta})\sqrt{N}\right)
\\
&\times
\exp
\left(
-\frac{(a-(b+1)p^2)(1+p^2)}{(a+b+1)p^2}\zeta\overline{\eta}(\zeta+\overline{\eta})
+
\frac{(1+p^2)a+(a-(b+1)p^2)p^2}{2(a+b+1)p^2}
(\zeta+\overline{\eta})^2
\right).
\end{split}
\end{align*}
Here, we adopt the sign $+$ if the outer edge, and we adopt the sign $-$ if the inner edge.
Together with these asymptotic behaviors, we have 
\begin{equation}
\mathfrak{p}^{-L}(1-\mathfrak{p})^{-n}\widehat{\omega}^{(n,L)}(z,\overline{w})
=
c(\zeta,\eta)\frac{1}{1+p^2}e^{-\frac{1}{2}(|\zeta|^2+|\eta|^2)+\zeta\overline{\eta}}\left(1+o(1)\right),
\label{LemAsym1}
\end{equation}
where $c(\zeta,\eta)$ is a a cocycle factor, which does not affect the value of determinant. 
Although we can explicitly write it, we do not need its explicit form here. 
We need the asymptotic expansion of the remainder term in the summation. 
Observe that by the Stirling formula, 
\begin{align*}
\binom{n+L}{N+L}
&=
\frac{\Gamma(n+L+1)}{\Gamma(N+L+1)\Gamma(n-N+1)}
=
\sqrt{\frac{a+b+1}{2\pi(a+1)bN}}
\left(\frac{(a+b+1)^{a+b+1}}{(a+1)^{a+1}b^b}\right)^N
\left(1+O(N^{-1})\right),\\
\binom{n+L}{N+1+L}
&=
\frac{\Gamma(n+L+1)}{\Gamma(N+L+2)\Gamma(n-N)}
=
\sqrt{\frac{a+b+1}{2\pi(a+1)bN}}
\left(\frac{(a+b+1)^{a+b+1}}{(a+1)^{a+1}b^b}\right)^N
\frac{b}{a+1}
\left(1+O(N^{-1})\right).
\end{align*}
Also, observe that 
\begin{align*}
\mathfrak{p}^{N+L}(1-\mathfrak{p})^{n-N}
&=
\begin{cases}
\left(\frac{(a+1)^{a+1}b^{b}}{(a+b+1)^{a+b+1}}\right)^Ne^{-\frac{1}{2}(\zeta+\overline{\eta})^2}\left(1+O(N^{-1/2})\right) & \text{(outer)},\\
\left(\frac{a^{a+1}(b+1)^{b}}{(a+b+1)^{a+b+1}}\right)^N\left(1+O(N^{-1})\right) & \text{(inner)},
\end{cases}
\\
\mathfrak{p}^{N+L+1}(1-\mathfrak{p})^{n-N-1}
&=
\begin{cases}
\left(\frac{(a+1)^{a+1}b^{b}}{(a+b+1)^{a+b+1}}\right)^N\frac{a+1}{b}e^{-\frac{1}{2}(\zeta+\overline{\eta})^2}\left(1+O(N^{-1/2})\right)  & \text{(outer)},\\
\left(\frac{a^{a+1}(b+1)^{b}}{(a+b+1)^{a+b+1}}\right)^N\frac{a}{b+1}\left(1+O(N^{-1})\right) & \text{(inner)}.
\end{cases}
\end{align*}
Hence, we obtain 
\begin{align*}
\binom{n+L}{N+L}\mathfrak{p}^{N+L}(1-\mathfrak{p})^{n-N}
&=
\begin{cases}
\sqrt{\frac{a+b+1}{2\pi(a+1)bN}}e^{-\frac{1}{2}(\zeta+\overline{\eta})^2}\left(1+O(N^{-1/2})\right) & \text{(outer)},\\
O(c^{-N}) & \text{(inner)},
\end{cases}
\\
\binom{n+L}{N+1+L}\mathfrak{p}^{N+L+1}(1-\mathfrak{p})^{n-N-1}
&=
\begin{cases}
\sqrt{\frac{a+b+1}{2\pi(a+1)bN}}e^{-\frac{1}{2}(\zeta+\overline{\eta})^2}\left(1+O(N^{-1/2})\right) & \text{(outer)},\\
O(c^{-N}) & \text{(inner)},
\end{cases}
\end{align*}
for some constant $c>1$. 
Furthermore, by Stirling formula, we have
\begin{align}
\label{LemAsym2}
\begin{split}
\frac{\frac{\Gamma(n+L+1)}{\Gamma(n+1)\Gamma(L)}}{nx-L}\widehat{\omega}^{(n,L)}(z,w)
=
\begin{cases}
O\left(e^{-\epsilon N}\right) & \text{(bulk)},\\
O\left(e^{-\epsilon N}\right) & \text{(outer)},\\
-
\frac{c(\zeta,\eta)}{1+p^2}\frac{e^{-\frac{1}{2}(\zeta+\overline{\eta})^2}e^{-\frac{1}{2}(|\zeta|^2+|\eta|^2)+\zeta\overline{\eta}}}{\sqrt{2\pi}(\chi+\overline{\chi})}
\left(1+O\left(\frac{1}{\sqrt{N}}\right)\right)
& \text{(inner)}.
\end{cases}
\end{split}
\end{align}
Here, $\epsilon$ is a positive constant independent of $N$. 
As a consequence, by combining these asymptotic behaviours, we obtain that for $k=0,1,2$,
\begin{align*}
\begin{split}
&\quad\widehat{q}_{N+k-1}^{(n,L)}(z\wbar|\lambda\overline{\lambda})
\widehat{\omega}^{(n,L)}(z,w)\\
&=
\begin{cases}
\frac{e^{-\frac{1}{2}(|\zeta|^2+|\eta|^2)+\zeta\overline{\eta}}}{1+p^2}\left(1+O(e^{-\epsilon N})\right)&
 \text{(bulk)},\\ 
\frac{e^{-\frac{1}{2}(|\zeta|^2+|\eta|^2)+\zeta\overline{\eta}}}{1+p^2}
\left(F(\zeta+\overline{\eta})+\frac{C_{\mathrm{out}}(\zeta,\eta)}{\sqrt{N}}+\sqrt{\frac{a+b+1}{2\pi(a+1)bN}}ke^{-\frac{1}{2}(\zeta+\overline{\eta})^2}
+O\left(\frac{1}{N}\right))
\right)
 & \text{(outer)},\\
\frac{e^{-\frac{1}{2}(|\zeta|^2+|\eta|^2)+\zeta\overline{\eta}}}{1+p^2}
\left(F(\zeta+\overline{\eta})-\frac{1}{\chi+\overline{\chi}}\frac{e^{-\frac{1}{2}(\zeta+\overline{\eta})^2}}{\sqrt{2\pi}}
+\frac{C_{\mathrm{out}}(\zeta,\eta)}{\sqrt{N}}+O\left(\frac{1}{N}\right))
\right)& \text{(inner)},
\end{cases}
\end{split}
\end{align*}
where we omitted the cocycle factor for the simplicity of the displays, and the convergence is uniform for $\zeta,\eta,\chi $ in compact subsets of $\C$. 
\par
\textbf{(II) Weakly non-unitary regime}:
under the setting of (II), by Taylor expansion, we have
\begin{align*}
\frac{N+L-1-(n+L)\mathfrak{p}}{\sqrt{(n+L)\mathfrak{p}(1-\mathfrak{p})}}
&=
-\Bigl( \zeta+\overline{\eta}-\frac{\rho}{\sqrt{2}}\Bigr)+ O(N^{-1}),
\\
\frac{L-(n+L)\mathfrak{p}}{\sqrt{(n+L)\mathfrak{p}(1-\mathfrak{p})}}
&=
-\Bigl( \zeta+\overline{\eta}+\frac{\rho}{\sqrt{2}}\Bigr)+ O(N^{-1}).
\end{align*}
Then, by Gaussian approximation, we have 
\begin{align*}
\begin{split}
\mathfrak{q}_{N-1}^{(n,L)}(\mathfrak{p})
&=
\P\bigl(L\leq X\leq N+L-1 \bigr) \\
&=
\P\Bigl(-\bigl( \eta+\overline{\zeta}+\frac{\rho}{\sqrt{2}}\bigr)\leq \mathcal{N}\leq -\bigl( \eta+\overline{\zeta}-\frac{\rho}{\sqrt{2}}\bigr) \Bigr) + O(N^{-1}),\quad \text{as $N\to\infty$}.
\end{split}
\end{align*}
By Taylor expansion, we have 
\begin{align}
\begin{split}
\label{LemAsym3}
\widehat{\omega}^{(n,L)}(\overline{z},w)
&=
c_N(\zeta,\eta)
2^{N-\frac{2N^2}{\rho^2}}
\frac{1}{2}e^{-\frac{1}{2}(|\zeta|^2+|\eta|^2)+\overline{\zeta}\eta}
e^{-\frac{1}{2}(\overline{\zeta}+\eta)^2-\frac{\sqrt{2}}{2}\rho(\overline{\zeta}+\eta)}
(1+O(N^{-1})), 
\end{split}
\end{align}
where $c_N(\zeta,\eta)=e^{-\sqrt{2}Ni(\im(\zeta)-\im(\eta))/\rho+i(\im\zeta^2-\im\eta^2)/2}$ is a co-cycle factor, which does not affect the value of the determinant. 
This has already been said. 
Hence, we omit the co-cycle factor in sequel. 
Note too that 
\[
\mathfrak{p}^{-L}(1-\mathfrak{p})^{-n}
=
2^{\frac{2N^2}{\rho^2}-N}e^{\frac{\sqrt{2}}{2}\rho(\overline{\zeta}+\eta)+\frac{1}{2}(\overline{\zeta}+\eta)^2}
\bigl(1+O(N^{-1})\bigr),
\]
\[
\mathfrak{p}^{N+L}(1-\mathfrak{p})^{n-N}
=
2^{-\frac{2N^2}{\rho^2}+N}e^{\frac{\sqrt{2}}{2}\rho(\overline{\zeta}+\eta)-\frac{1}{2}(\overline{\zeta}+\eta)^2}
\bigl(1+O(N^{-1})\bigr),
\]
and 
\[
\mathfrak{p}^{N+L+1}(1-\mathfrak{p})^{n-N-1}
=
2^{-\frac{2N^2}{\rho^2}+N}e^{\frac{\sqrt{2}}{2}\rho(\overline{\zeta}+\eta)-\frac{1}{2}(\overline{\zeta}+\eta)^2}
\bigl(1+O(N^{-1})\bigr),
\]
which give that for $k=1,2$, 
\begin{align*}
&\quad 
\widehat{\omega}^{(n,L)}(\overline{z},w)
\mathfrak{p}^{-L}(1-\mathfrak{p})^{-n}
\mathfrak{p}^{N+L+(k-1)}(1-\mathfrak{p})^{n-N-(k-1)}
\\
&=
2^{N-\frac{2N^2}{\rho^2}}
\frac{1}{2}e^{-\frac{1}{2}(|\zeta|^2+|\eta|^2)+\overline{\zeta}\eta}
e^{-\frac{1}{2}(\overline{\zeta}+\eta)^2+\frac{\sqrt{2}}{2}\rho(\overline{\zeta}+\eta)}
\bigl(1+O(N^{-1})\bigr).
\end{align*}
Also, by Stirling formula, we have 
\begin{align*}
\binom{n+L}{N+L}&=
2^{\frac{2N^2}{\rho^2}-N}e^{-\frac{\rho^2}{4}}
\frac{\rho}{\sqrt{\pi} N}
\Bigl( 1+\frac{2\rho^2-\rho^4}{8N}+O(N^{-2}) \Bigr),
\\
\binom{n+L}{N+L+1}&=
2^{\frac{2N^2}{\rho^2}-N}e^{-\frac{\rho^2}{4}}
\frac{\rho}{\sqrt{\pi} N}
\Bigl( 1-\frac{6\rho^2+\rho^4}{8N}+O(N^{-2}) \Bigr). 
\end{align*}
Therefore, we have that for $k=1,2$, 
\begin{align*}
&\quad
\mathfrak{p}^{-L}(1-\mathfrak{p})^{-n}
\widehat{\omega}^{(n,L)}(\overline{z},w)
\sum_{j=1}^k\binom{n+L}{N+L+j-1}
\mathfrak{p}^{N+L+j-1}(1-\mathfrak{p})^{n-N-(j-1)}
\\
&=
k
\frac{\rho}{2\sqrt{\pi}N}
e^{-\frac{1}{2}(|\zeta|^2+|\eta|^2)+\overline{\zeta}\eta}e^{-\frac{1}{2}\bigl(\overline{\zeta}+\eta-\frac{\rho}{\sqrt{2}}\bigr)^2}
\bigl(1+O(N^{-1})\bigr).
\end{align*}
Here, note that 
\[
\frac{1}{nx-L}\frac{\Gamma(L+n+1)}{\Gamma(n+1)\Gamma(L)}
\widehat{\omega}^{(n,L)}(\overline{z},w)
=
\frac{e^{-\frac{1}{2}(|\zeta|^2+|\eta|^2)+\overline{\zeta}\eta}}{2\sqrt{2\pi}(\chi+\overline{\chi}+\frac{\rho}{\sqrt{2}})}
e^{-\frac{1}{2}\bigl(\overline{\zeta}+\eta+\frac{\rho}{\sqrt{2}} \bigr)^2}\bigl( 1+ O(N^{-1}) \bigr). 
\]
Together with asymptotic expansions so far, we obtain 
\begin{align*}
&\quad
\widehat{q}_{N-1+k}^{(n,L)}(\overline{z}w|\lambda\overline{\lambda})
\widehat{\omega}^{(n,L)}(\overline{z},w)
\\
&=
\frac{e^{-\frac{1}{2}(|\zeta|^2+|\eta|^2)+\overline{\zeta}\eta}}{2}
\Bigl( 
\mathcal{J}_{\rho}(\overline{\zeta},\eta|\overline{\chi},\chi)
+
k
\frac{\rho}{N}
\frac{e^{-\frac{1}{2}\bigl(\overline{\zeta}+\eta-\frac{\rho}{\sqrt{2}}\bigr)^2}}{\sqrt{\pi}}
+
\frac{C_{\mathrm{w}}(\zeta,\eta)}{N}+ O(N^{-2})
\Bigr),
\end{align*}
as $N\to\infty$,
where some constant $C_{\mathrm{w}}(\zeta,\eta)$ depends on $\zeta,\eta,\rho$. 
Here, the convergence is uniform for $\zeta,\eta,\chi$ in compact subsets of $\C$.
\par
\textbf{(III) Singular origin regime}:
under the setting of (III), by Taylor expansion,
\[
z\wbar=\frac{\zeta\overline{\eta}}{1+b}\frac{1}{N}+O\left(\frac{1}{N^2}\right)
\]
and by Poisson approximation, we have
\[
\P(L\leq X\leq N+k+L-1)=P(L,\zeta\overline{\eta})\left(1+o(1)\right),\quad\text{as $N\to\infty$}, 
\]
where 
\[
P(c,z)=\frac{1}{\Gamma(c)}\int_0^z t^{c-1}e^{-t}dt=e^{-z}z^cE_{1,1+c}(z)\quad (c>0).
\]
On the other hand, since
\[
\mathfrak{p}^{-L}\left(1-\mathfrak{p}\right)^{-n}
\widehat{\omega}_N^{(n,L)}(z,w)
=
c(\zeta,\eta)e^{-\frac{1}{2}(|\zeta|^2+|\eta|^2)+\zeta\overline{\eta}},
\]
where $c(\zeta,\eta)$ is a co-cycle factor, and hence, by simple computations, we have 
\[
\widehat{q}_{N+k-1}^{(n,L)}(z\wbar)
\widehat{\omega}_N^{(n,L)}(z,w)
=
\frac{1}{\chi\overline{\chi}-L}(\zeta\overline{\eta})^L\mathcal{E}_{1,L}\left(\zeta\overline{\eta}|\overline{\chi}\chi\right)e^{-\frac{1}{2}(|\zeta|^2+|\eta|^2)}\left(1+o(1)\right),
\]
uniformly for $\zeta,\eta,\chi$ in a compact subset of $\C$ and we omitted the cocycle factor. 
\end{proof}
\begin{rem}
Our proof followed the probabilistic argument by \cite{SFv3}. 
Our claim can be also proved by using the asymptotics for the incomplete beta function in \cite{TEM79}. 
\end{rem}
With help of Lemma~\ref{Lem_Asym}, we can complete the proof of Theorem~\ref{MainResult}.
Now, we shall finish the proof for each case. 
\subsection{Proof of Theorem~\ref{MainResult} of the strongly non-unitary regime}
In this subsection, we assume that $L=aN$ and $n=(b+1)N$ with $a,b\geq0$. 
\begin{proof}[Proof of Theorem~\ref{MainResult} for the bulk case in strongly non-unitary regime]
Through the proof of Theorem~\ref{MainResult} for the bulk case in strongly non-unitary regime, 
let 
\[
z=p+\frac{\zeta}{\sqrt{N\delta_N(p)}},\quad
w=p+\frac{\eta}{\sqrt{N\delta_N(p)}},\quad
\lambda=p+\frac{\chi}{\sqrt{N\delta_N(p)}}.
\]
First, note that 
\[
N\delta_N(p)=N\frac{a+b+1}{(1+|p|^2)^2}(1+O(N^{-1})). 
\]
Second, note that 
\[
\varpi^{(n,L)}(z,\overline{z},w,\overline{w}|\lambda,\overline{\lambda})
=
\frac{(1+|p|^2)^2}{(a+b+1)N}
\varpi(\zeta,\overline{\zeta},\eta,\overline{\eta}|\chi,\overline{\chi}),
(1+O(N^{-1}))
\]
and by \eqref{LemAsym1}, we have 
\[
C_{K}(z,w)=e^{\frac{1}{2}(|\zeta|^2+|\eta|^2)-\overline{\zeta}\eta}(1+O(N^{-1})), 
\]
up to a co-cycle factor. 
Notice also that 
\[
\widehat{g}_N^{(n,L)}(\lambda\overline{\lambda})\widehat{\omega}^{(n,L)}(\lambda,\overline{\lambda})
=
N\frac{a+1-b|p|^2}{1+|p|^2}\bigl( 1 + O(N^{-1}) \bigr).
\]
Then, we have 
\[
\lim_{N\to\infty}
\frac{1}{N}\frac{1}{N\delta_N(p)}
\frac{n\bigl(z\overline{z}-\frac{L}{n} \bigr)}{z\overline{z}(1+z\overline{z})}
g_{N-1}^{(n,L)}(z\overline{z})\widehat{\omega}^{(n,L)}(\overline{z},z)
=
\frac{b(b+1)}{a+b+1}\frac{\Bigl( |p|^2-\frac{a}{b+1} \Bigr)\Bigl( \frac{a+1}{b}-|p|^2 \Bigr)}{|p|^2}.
\]
Since 
\begin{align*}
&\quad
(L+N+1)Q_{N+1}^{(n,L)}(\overline{z},w,\lambda)-\lambda\overline{\lambda}(n-N-1)Q_{N+1}^{(n,L)}(\overline{z},w,\lambda)
\\
&=
N\frac{a+1-b|p|^2}{(1+|p|^2)^2}e^{-\frac{1}{2}(|\zeta|^2+|\eta|^2)+\overline{\zeta}\eta-(\overline{\zeta}-\overline{\eta})(\eta-\chi)}
\bigl( 
1-e^{(\overline{\zeta}-\overline{\chi})(\eta-\chi)}
\bigr)\bigl( 1+ O(N^{-1}) \bigr), 
\end{align*}
we have 
\begin{align*}
C_K(z,w)\mathfrak{H}_N^{(n,L)}(z,w,\lambda)
&=
(N\delta_N(p))^2 \frac{1-e^{(\overline{\zeta}-\overline{\chi})(\eta-\chi)}}{(\overline{\zeta}-\overline{\chi})^2(\eta-\chi)^2}
e^{-(\overline{\zeta}-\overline{\chi})(\eta-\chi)}
\bigl( 1+ O(N^{-1}) \bigr).
\end{align*}
On the other hand, similarly, we have 
\[
C_K(z,w)\mathfrak{F}_N^{(n,L)}(z,w,\lambda)
=
\frac{(N\delta_N(p))^2}{(\overline{\zeta}-\overline{\chi})(\eta-\chi)}\bigl( 1+ O(N^{-1}) \bigr).
\]
Hence, we have 
\begin{align*}
K_{1,1}^{(N)}(z,\overline{z},w,\overline{w}|\lambda,\overline{\lambda})
&=
N\delta_N(p) \frac{1+(\overline{\zeta}-\overline{\chi})(\eta-\chi)e^{-(\overline{\zeta}-\overline{\chi})(\eta-\chi)}-e^{-(\overline{\zeta}-\overline{\chi})(\eta-\chi)}}{(\overline{\zeta}-\overline{\chi})^2(\eta-\chi)^2}
\\
&\times
\varpi(\zeta,\overline{\zeta},\eta,\overline{\eta}|\chi,\overline{\chi})
e^{-(\overline{\zeta}-\overline{\chi})(\eta-\chi)}\bigl( 1+ O(N^{-1}) \bigr),
\end{align*}
which gives 
\[
\lim_{N\to\infty}\frac{1}{N\delta_N(p)}
K_{1,1}^{(N)}(z,\overline{z},w,\overline{w}|\lambda,\overline{\lambda})
=
K_{1,1}^{(\mathbf{b})}(z,\overline{z},w,\overline{w}|\lambda,\overline{\lambda}),
\]
uniformly for $\zeta,\eta,\chi$ in compact subsets of $\C$. 
This completes the proof of \eqref{D11BULK}.
By decoupling Lemma~\ref{Decoupling}, it suffices to consider the front factor for \eqref{D12BULK}.
However, since the exponential factor in \eqref{BulkWeight} becomes $e^{(\overline{\zeta_1}-\overline{\zeta_2})(\zeta_1-\zeta_2)}$, from the proof of the diagonal case and by \eqref{D12ver1}, we have 
\begin{align*}
&
\quad -\lim_{N\to\infty} \frac{1}{N}\frac{1}{(N\delta_N(p))^2}
\frac{n(z\overline{w}-\frac{L}{n})}{z\overline{w}(1+z\overline{w})}
\widehat{g}_{N-1}^{(n,L)}(z\overline{w})\widehat{\omega}^{(n,L)}(\overline{w},z)
\mathcal{K}^{(N-1)}(\overline{z},w|z,\overline{w})\widehat{\omega}^{(n,L)}(\overline{z},w).
\\
&
=
-
\frac{\frac{b(b+1)}{a+b+1}\Bigl( |p|^2-\frac{a}{b+1} \Bigr)\Bigl( \frac{a+1}{b}-|p|^2 \Bigr)}{|p|^2}
\mathcal{K}^{(\mathbf{b})}(\overline{\zeta},\eta|\zeta,\overline{\eta}).
\end{align*}
This completes the proof of \eqref{D12BULK}
\end{proof}

\begin{proof}[Proof of Theorem~\ref{MainResult} for the edge case in strongly non-unitary regime]
Similar to the proof of Theorem~\ref{MainResult} for the bulk case in strongly non-unitary regime, 
for $\theta\in[0,2\pi)$, let 
\[
z=e^{i\theta}\Bigl(p+\mathfrak{s}\frac{\zeta}{\sqrt{N\delta_N(p)}}\Bigr),\quad
w=e^{i\theta}\Bigl(p+\mathfrak{s}\frac{\eta}{\sqrt{N\delta_N(p)}}\Bigr),\quad
\lambda=e^{i\theta}\Bigl(p+\mathfrak{s}\frac{\chi}{\sqrt{N\delta_N(p)}}\Bigr).
\]
Here, if we consider the outer edge case, then $p=\sqrt{(a+1)/b}$, and if we consider the inner edge case, then $p=\sqrt{a/(b+1)}$. 
Also, if we consider the outer edge case, then $\mathfrak{s}=1$, and if we consider the inner edge case, then $\mathfrak{s}=-1$. 
First, note that 
\[
\varpi^{(n,L)}(z,\overline{z},w,\overline{w}|\lambda,\overline{\lambda})
=
\frac{(1+|p|^2)^2}{(a+b+1)N}
\varpi(\zeta,\overline{\zeta},\eta,\overline{\eta}|\chi,\overline{\chi})
(1+O(N^{-1}))
\]
and 
\[
\widehat{g}_N^{(n,L)}(\lambda\overline{\lambda})
\widehat{\omega}^{(n,L)}(\overline{\lambda},\lambda)
=
\begin{cases}
\frac{\sqrt{N}}{\sqrt{2\pi}}\sqrt{\frac{(a+1)b}{a+b+1}}\mathcal{F}(\chi+\overline{\chi})(1+O(N^{-1/2})),  & \text{if $\mathfrak{s}=1$},\\
-\frac{N}{\sqrt{2\pi}(\chi+\overline{\chi})}\mathcal{F}(\chi+\overline{\chi})(1+O(N^{-1/2})),  & \text{if $\mathfrak{s}=-1$},
\end{cases}
\]
where $\mathcal{F}(x)$ is given by \eqref{EdgeF2}. 
Then, we have 
\[
\lim_{N\to\infty}
\frac{1}{\sqrt{N}}\frac{1}{N\delta_N(p)}
\frac{n\bigl(z\overline{z}-\frac{L}{n} \bigr)}{z\overline{z}(1+z\overline{z})}
g_{N-1}^{(n,L)}(z\overline{z})\widehat{\omega}^{(n,L)}(\overline{z},z)
=
\begin{cases}
\sqrt{\frac{a+b+1}{2\pi(a+1)b}}\mathcal{F}(\zeta+\overline{\zeta})(1+O(N^{-1/2})),  & \text{if $\mathfrak{s}=1$},\\
\sqrt{\frac{a+b+1}{2\pi a(b+1)}}\mathcal{F}(\zeta+\overline{\zeta})(1+O(N^{-1/2}))  & \text{if $\mathfrak{s}=-1$}. 
\end{cases}
\]
\par
\textbf{(1) outer edge case}: First, we consider the outer edge case. 
By Lemma~\ref{Lem_Asym}, we have 
\begin{align*}
&\quad
(N+L+1)Q_{N+1}^{(n,L)}(z,w,\lambda)-\lambda\overline{\lambda}(n-N-1)Q_{N}^{(n,L)}(z,w,\lambda)
\\
&=
\frac{\sqrt{N}\sqrt{a+1}b^{3/2}}{\sqrt{2\pi}(a+b+1)^{3/2}}
H_1(\zeta,\eta,\chi)
e^{-\frac{1}{2}(|\zeta|^2+|\eta|^2)+\overline{\zeta}\eta}
(1+O(N^{-1/2})),
\end{align*}
where we set 
\begin{align*}
H_1(\zeta,\eta,\chi)
&=
\sqrt{2\pi}(\chi+\overline{\chi})F(\overline{\zeta}+\eta)F(\overline{\chi}+\chi)
-F(\overline{\zeta}+\eta)e^{-\frac{1}{2}(\chi+\overline{\chi})^2}-F(\overline{\chi}+\chi)e^{-\frac{1}{2}(\overline{\zeta}+\eta)^2}
\\
&
+e^{-(\overline{\zeta}-\overline{\chi})(\eta-\chi)}
\Bigl( 
F(\overline{\zeta}+\chi)e^{-\frac{1}{2}(\overline{\chi}+\eta)^2}+F(\overline{\chi}+\eta)e^{-\frac{1}{2}(\overline{\zeta}+\chi)^2}
-\sqrt{2\pi}(\chi+\overline{\chi})F(\overline{\zeta}+\chi)F(\overline{\chi}+\eta)
\Bigr).
\end{align*} 
Then, we have 
\[
C_{K}(z,w)\mathfrak{H}_N^{(n,L)}(z,w,\lambda)
\varpi^{(n,L)}(z,\overline{z},w,\overline{w}|\lambda,\overline{\lambda})
=
N\delta_N(p) 
\frac{H_1(\zeta,\eta,\chi)
\varpi(\zeta,\overline{\zeta},\eta,\overline{\eta}|\chi,\overline{\chi})}{(\overline{\zeta}-\overline{\chi})^2(\eta-\chi)^2\mathcal{F}(\chi+\overline{\chi})}
(1+O(N^{-1/2})). 
\]
Similarly, we have 
\[
C_{K}(z,w)\mathfrak{I}_N^{(n,L)}(z,w,\lambda)
\varpi^{(n,L)}(z,\overline{z},w,\overline{w}|\lambda,\overline{\lambda})
=
N\delta_N(p)\frac{F(\overline{\zeta}+\eta)\varpi(\zeta,\overline{\zeta},\eta,\overline{\eta}|\chi,\overline{\chi})}{(\overline{\zeta}-\overline{\chi})(\eta-\chi)}
(1+O(N^{-1/2})). 
\]
For \eqref{Part2N} and \eqref{Part3N}, it is straightforward to see that 
\[
C_{K}(z,w)\mathfrak{II}_N^{(n,L)}(z,w,\lambda)
\varpi^{(n,L)}(z,\overline{z},w,\overline{w}|\lambda,\overline{\lambda})
=o(N\delta_N(p)),
\]
and 
\[
C_{K}(z,w)\mathfrak{III}_N^{(n,L)}(z,w,\lambda)
\varpi^{(n,L)}(z,\overline{z},w,\overline{w}|\lambda,\overline{\lambda})
=o(N\delta_N(p)), 
\]
as $N\to\infty$. 
Here, we recall that $\partial_xF(a+x)=e^{-\frac{1}{2}(a+x)^2}/\sqrt{2\pi}$, and 
\begin{align*}
&\quad e^{-\frac{1}{2}a^2}\frac{d}{dx}
\Bigl[
e^{\frac{(a+x)^2}{2}}
\Bigl(e^{-f}F(b+x)F(c+x)-F(d+x)F(a+x)+fF(d)F(a+x) \Bigr)
\Bigr]\Bigr|_{x=0}
\\
&=
\frac{1}{\sqrt{2\pi}}
\Bigl[
\sqrt{2\pi}a\Bigl(e^{-f}F(b)F(c)-F(d)F(a)+fF(d)F(a) \Bigr)
\\
&
-
\Bigl( 
e^{-f}e^{-\frac{1}{2}b^2}F(c)+e^{-f}e^{-\frac{1}{2}c^2}F(b)
-e^{-\frac{1}{2}d^2}F(a)-e^{-\frac{1}{2}a^2}F(d)
+fe^{-\frac{1}{2}a^2}F(d)
\Bigr)
\Bigr]. 
\end{align*}
Further we note that 
\begin{align*}
&\quad
H_1(\zeta,\eta+\chi)+(\overline{\zeta}-\overline{\chi})(\eta-\chi)F(\overline{\zeta}+\eta)\mathcal{F}(\overline{\chi}+\chi)
\\
&=
\sqrt{2\pi}(\chi+\overline{\chi})F(\overline{\zeta}+\eta)F(\overline{\chi}+\chi)
-e^{-(\overline{\zeta}-\overline{\chi})(\eta-\chi)}\sqrt{2\pi}(\chi+\overline{\chi})F(\overline{\zeta}+\chi)F(\overline{\chi}+\eta)
\\
&
-\sqrt{2\pi}(\overline{\zeta}-\overline{\chi})(\eta-\chi)(\chi+\overline{\chi})F(\overline{\zeta}+\eta)F(\overline{\chi}+\chi)
-F(\overline{\zeta}+\eta)e^{-\frac{1}{2}(\chi+\overline{\chi})^2}-F(\overline{\chi}+\chi)e^{-\frac{1}{2}(\overline{\zeta}+\eta)^2}
\\
&
+e^{-(\overline{\zeta}-\overline{\chi})(\eta-\chi)}
F(\overline{\zeta}+\chi)e^{-\frac{1}{2}(\overline{\chi}+\eta)^2}
+e^{-(\overline{\zeta}-\overline{\chi})(\eta-\chi)}F(\overline{\chi}+\eta)e^{-\frac{1}{2}(\overline{\zeta}+\chi)^2}
+(\overline{\zeta}-\overline{\chi})(\eta-\chi)e^{-\frac{1}{2}(\chi+\overline{\chi})^2}F(\overline{\zeta}+\eta)
\\
&=
-\sqrt{2\pi}(\chi+\overline{\chi})
\bigl(
e^{-(\overline{\zeta}-\overline{\chi})(\eta-\chi)}F(\overline{\zeta}+\chi)F(\overline{\chi}+\eta)
-F(\overline{\zeta}+\eta)F(\overline{\chi}+\chi) 
+(\overline{\zeta}-\overline{\chi})(\eta-\chi)F(\overline{\zeta}+\eta)F(\overline{\chi}+\chi)
 \bigr)\\
 &
 +\Bigl( e^{-(\overline{\zeta}-\overline{\chi})(\eta-\chi)}F(\overline{\zeta}+\chi)e^{-\frac{1}{2}(\overline{\chi}+\eta)^2}
 +e^{-(\overline{\zeta}-\overline{\chi})(\eta-\chi)}F(\overline{\chi}+\eta)e^{-\frac{1}{2}(\overline{\zeta}+\chi)^2}
 \\
 &
 -F(\overline{\zeta}+\eta)e^{-\frac{1}{2}(\chi+\overline{\chi})^2}
 -F(\overline{\chi}+\chi)e^{-\frac{1}{2}(\overline{\zeta}+\eta)^2}
+(\overline{\zeta}-\overline{\chi})(\eta-\chi)e^{-\frac{1}{2}(\chi+\overline{\chi})^2}F(\overline{\zeta}+\eta)
\Bigr). 
\end{align*} 
Therefore, we have 
\begin{align*}
&\quad
H_1(\zeta,\eta+\chi)+(\overline{\zeta}-\overline{\chi})(\eta-\chi)F(\overline{\zeta}+\eta)\mathcal{F}(\overline{\chi}+\chi)
\\
&=
-\sqrt{2\pi}
e^{-\frac{1}{2}a^2}\frac{d}{dx}
\Bigl[
e^{\frac{(a+x)^2}{2}}
\Bigl(e^{-f}F(b+x)F(c+x)-F(d+x)F(a+x)+fF(d)F(a+x) \Bigr)
\Bigr]\Bigr|_{x=0}.
\end{align*}
where $a=\overline{\chi}+\chi,b=\overline{\zeta}+\chi,c=\overline{\chi}+\eta,d=\overline{\zeta}+\eta,f=(\overline{\zeta}-\overline{\chi})(\eta-\chi)$. 
This completes the proof for the outer edge case. 
\par
\textbf{(2) inner edge case}: Next, we consider the outer edge case. 
By Lemma~\ref{Lem_Asym}, we have 
\begin{align*}
&\quad
(N+L+1)Q_{N+1}^{(n,L)}(z,w,\lambda)-\lambda\overline{\lambda}(n-N-1)Q_{N}^{(n,L)}(z,w,\lambda)
\\
&=
N\frac{b+1}{a+b+1}\frac{e^{-\frac{1}{2}(|\zeta|^2+|\eta|^2)+\overline{\zeta}\eta}}{\sqrt{2\pi}(\overline{\chi}+\chi)}
H_2(\zeta,\eta,\chi)\bigl( 1+ O(N^{-1/2}) \bigr),
\end{align*}
where we set 
\begin{align*}
H_2(\zeta,\eta,\chi)
&=
\sqrt{2\pi}(\overline{\chi}+\chi)
\bigl(e^{-(\overline{\zeta}-\overline{\chi})(\eta-\chi)}F(\overline{\zeta}+\chi)F(\overline{\chi}+\eta) - F(\overline{\chi}+\chi)F(\overline{\zeta}+\eta) \bigr)
\\
&
+e^{-\frac{1}{2}(\overline{\chi}+\chi)^2}F(\overline{\zeta}+\eta)
+e^{-\frac{1}{2}(\overline{\zeta}+\eta)^2}F(\overline{\chi}+\chi)
\\
&
-e^{-(\overline{\zeta}-\overline{\chi})(\eta-\chi)}F(\overline{\zeta}+\chi)e^{-\frac{1}{2}(\overline{\chi}+\eta)^2}
-e^{-(\overline{\zeta}-\overline{\chi})(\eta-\chi)}F(\overline{\chi}+\eta)e^{-\frac{1}{2}(\overline{\zeta}+\chi)^2}. 
\end{align*}
Then, we have 
\begin{align*}
&\quad 
C_K(z,w)\mathfrak{H}_N^{(n,L)}(z,w,\lambda)
\varpi^{(n,L)}(z,\overline{z},w,\overline{w}|\lambda,\overline{\lambda})
=
-N\delta_N(p) \frac{H_2(\zeta,\eta,\chi)\varpi(\zeta,\overline{\zeta},\eta,\overline{\eta}|\chi,\overline{\chi})}{(\overline{\zeta}-\overline{\chi})^2(\eta-\chi)^2\mathcal{F}(\overline{\chi}+\chi)}
\bigl( 1+ O(N^{-1/2})\bigr). 
\end{align*}
Similarly, we have 
\[
C_{K}(z,w)\mathfrak{I}_N^{(n,L)}(z,w,\lambda)
\varpi^{(n,L)}(z,\overline{z},w,\overline{w}|\lambda,\overline{\lambda})
=
N\delta_N(p)\frac{F(\overline{\zeta}+\eta)\varpi(\zeta,\overline{\zeta},\eta,\overline{\eta}|\chi,\overline{\chi})}{(\overline{\zeta}-\overline{\chi})(\eta-\chi)}
\bigl( 1+ O(N^{-1/2})\bigr). 
\]
As in the outer edge case, for \eqref{Part2N} and \eqref{Part3N}, it is straightforward to see that as $N\to\infty$, 
\[
C_{K}(z,w)\mathfrak{II}_N^{(n,L)}(z,w,\lambda)
\varpi^{(n,L)}(z,\overline{z},w,\overline{w}|\lambda,\overline{\lambda})
=o(N\delta_N(p)),
\]
and 
\[
C_{K}(z,w)\mathfrak{III}_N^{(n,L)}(z,w,\lambda)
\varpi^{(n,L)}(z,\overline{z},w,\overline{w}|\lambda,\overline{\lambda})
=o(N\delta_N(p)). 
\]
Here, we note that 
\begin{align*}
&\quad-H_2(\zeta,\eta,\chi)+(\overline{\zeta}-\overline{\chi})(\eta-\chi)F(\overline{\zeta}+\eta)\mathcal{F}(\overline{\chi}+\chi)
\\
&=
-\sqrt{2\pi}(\overline{\chi}+\chi)
\Bigl(
e^{-(\overline{\zeta}-\overline{\chi})(\eta-\chi)}F(\overline{\zeta}+\chi)F(\overline{\chi}+\eta)
 -F(\overline{\chi}+\chi)F(\overline{\zeta}+\eta) 
+(\overline{\zeta}-\overline{\chi})(\eta-\chi)F(\overline{\chi}+\chi)F(\overline{\zeta}+\eta)
 \Bigr)
 \\
&
+e^{-(\overline{\zeta}-\overline{\chi})(\eta-\chi)}F(\overline{\zeta}+\chi)e^{-\frac{1}{2}(\overline{\chi}+\eta)^2}
+e^{-(\overline{\zeta}-\overline{\chi})(\eta-\chi)}F(\overline{\chi}+\eta)e^{-\frac{1}{2}(\overline{\zeta}+\chi)^2}
\\
&
-F(\overline{\zeta}+\eta)e^{-\frac{1}{2}(\overline{\chi}+\chi)^2}
-F(\overline{\chi}+\chi)e^{-\frac{1}{2}(\overline{\zeta}+\eta)^2}
+(\overline{\zeta}-\overline{\chi})(\eta-\chi)F(\overline{\zeta}+\eta)e^{-\frac{1}{2}(\chi+\overline{\chi})^2}
\\
&=
-\sqrt{2\pi}
e^{-\frac{1}{2}a^2}\frac{d}{dx}
\Bigl[
e^{\frac{(a+x)^2}{2}}
\Bigl(e^{-f}F(b+x)F(c+x)-F(d+x)F(a+x)+fF(d)F(a+x) \Bigr)
\Bigr]\Bigr|_{x=0}.
\end{align*}
where $a=\overline{\chi}+\chi,b=\overline{\zeta}+\chi,c=\overline{\chi}+\eta,d=\overline{\zeta}+\eta,f=(\overline{\zeta}-\overline{\chi})(\eta-\chi)$. 
This completes the proof of \eqref{D11BULK} and \eqref{ThmEdge1}.
Similarly, it suffices to consider the front factor for \eqref{D12EDGE}.
From the proof of the diagonal case and by \eqref{D12ver1}, 
\begin{align*}
&
\quad -\lim_{N\to\infty} \frac{1}{\sqrt{N}}\frac{1}{(N\delta_N(p))^2}
\frac{n(z\overline{w}-\frac{L}{n})}{z\overline{w}(1+z\overline{w})}
\widehat{g}_{N-1}^{(n,L)}(z\overline{w})\widehat{\omega}^{(n,L)}(\overline{w},z)
\mathcal{K}^{(N-1)}(\overline{z},w|z,\overline{w})\widehat{\omega}^{(n,L)}(\overline{z},w).
\\
&
=
-
\mathfrak{c}_{\mathfrak{s}}e^{-|\zeta-\eta|^2}\mathcal{F}(\zeta+\overline{\eta})
\frac{H(\overline{\eta}+\zeta,\overline{\zeta}+\zeta,\overline{\eta}+\eta,\overline{\eta}+\zeta,-(\overline{\zeta}-\overline{\eta})(\zeta-\eta))}{(\overline{\zeta}-\overline{\eta})^2(\zeta-\eta)^2}.
\end{align*}
This completes the proof of \eqref{D12EDGE}. 
\end{proof}

\subsection{Proof of Theorem~\ref{MainResult} in the weakly non-unitary regime}
In this section, we assume that $L=N^2/\rho^2-N$ and $n=N^2/\rho^2$. 
\begin{proof}[Proof of Theorem~\ref{MainResult} in the weakly non-unitary regime]
Through the proof of Theorem~\ref{MainResult} in the weakly non-unitary regime, for $\theta\in[0,2\pi)$, we set 
\[
z=e^{i\theta}\Bigl(1+\frac{\zeta}{\sqrt{N\delta_N(1)}} \Bigr),\quad 
w=e^{i\theta}\Bigl(1+\frac{\eta}{\sqrt{N\delta_N(1)}} \Bigr),\quad
\lambda=e^{i\theta}\Bigl(1+\frac{\chi}{\sqrt{N\delta_N(1)}} \Bigr).
\]
Note that 
\[
N\delta_N(1)=\frac{N^2}{2\rho^2}\Bigl( 1 -\frac{\rho^2}{2N} + \frac{\rho^2}{2N^2}\Bigr).
\]
Notice also that 
\[
\varpi^{(n,L)}(z,\overline{z},w,\overline{w}|\lambda,\overline{\lambda})
=
\frac{1}{N\delta_N(1)}
\varpi(\zeta,\overline{\zeta},\eta,\overline{\eta}|\chi,\overline{\chi})
\bigl( 1 + O(N^{-1}) \bigr), 
\]
and 
\[
\widehat{g}_N^{(n,L)}(\lambda\overline{\lambda})
\widehat{\omega}^{(n,L)}(\overline{\lambda},\lambda)
=
\frac{N}{\sqrt{2}\rho} \frac{1}{\chi+\overline{\chi}+\frac{\rho}{\sqrt{2}}}
\mathcal{L}_{\rho}(\chi+\overline{\chi})
\bigl( 1 + O(N^{-1}) \bigr).
\]
Then, we have 
\[
\lim_{N\to\infty}\frac{1}{N\delta_N(1)}
\frac{n(\lambda\overline{\lambda}-L/n)}{\lambda\overline{\lambda}(1+\lambda\overline{\lambda})}
\widehat{g}_N^{(n,L)}(\lambda\overline{\lambda})\widehat{\omega}^{(n,L)}(\lambda,\overline{\lambda})
=
\mathcal{L}_{\rho}(\chi+\overline{\chi}). 
\]
Since 
\begin{align*}
&\quad Q_{N+1}^{(n,L)}(z,w,\lambda)-Q_N(z,w,\lambda)
\\
&=
\frac{\rho}{N}
\frac{e^{-\frac{1}{2}(|\zeta|^2+|\eta|^2)+\overline{\zeta}\eta}}{2^2\sqrt{\pi}}
e^{-(\overline{\zeta}-\overline{\chi})(\eta-\chi)}
\\
&\quad 
\times
\Bigl\{
\mathcal{J}_{\rho}(\overline{\zeta},\chi|\overline{\chi},\chi)e^{-\frac{1}{2}(\overline{\chi}+\eta-\frac{\rho}{\sqrt{2}})^2}
+
\mathcal{J}_{\rho}(\overline{\chi},\eta|\overline{\chi},\chi)e^{-\frac{1}{2}(\overline{\zeta}+\chi-\frac{\rho}{\sqrt{2}})^2}
\\
&
\quad \quad
-e^{(\overline{\zeta}-\overline{\chi})(\eta-\chi)}\mathcal{J}_{\rho}(\overline{\zeta},\eta|\overline{\eta},\chi)e^{-\frac{1}{2}(\overline{\chi}+\chi-\frac{\rho}{\sqrt{2}})^2}
-e^{(\overline{\zeta}-\overline{\chi})(\eta-\chi)}\mathcal{J}_{\rho}(\overline{\chi},\chi|\overline{\chi},\chi)e^{-\frac{1}{2}(\overline{\zeta}+\eta-\frac{\rho}{\sqrt{2}})^2}
\Bigr\}
\bigl( 1 + O(N^{-1}) \bigr), 
\end{align*}
we have 
\begin{align*}
&\quad
(N+L+1)Q_{N+1}^{(n,L)}(z,w,\lambda)-\lambda\overline{\lambda}(n-N-1)Q_{N}^{(n,L)}(z,w,\lambda)
\\
&=
\frac{\sqrt{2}N}{\rho}
\frac{e^{-\frac{1}{2}(|\zeta|^2+|\eta|^2)+\overline{\zeta}\eta}}{2^2} 
\frac{e^{-(\overline{\zeta}-\overline{\chi})(\eta-\chi)}}{\chi+\overline{\chi}+\frac{\rho}{\sqrt{2}}}
H_{3}(\zeta,\eta,\chi)\bigl( 1 + O(N^{-1}) \bigr), 
\end{align*}
where
\begin{align*}
H_3(\zeta,\eta,\chi)
&=
\mathcal{B}_{\rho}(\overline{\chi}+\chi,\overline{\zeta}+\chi,\overline{\chi}+\eta)
+
\mathcal{B}_{\rho}(\overline{\chi}+\chi,\overline{\chi}+\eta,\overline{\zeta}+\chi)
\\
&
-e^{(\overline{\zeta}-\overline{\chi})(\eta-\chi)}\mathcal{B}_{\rho}(\overline{\chi}+\chi,\overline{\zeta}+\eta,\overline{\chi}+\chi)
-e^{(\overline{\zeta}-\overline{\chi})(\eta-\chi)}\mathcal{B}_{\rho}(\overline{\chi}+\chi,\overline{\chi}+\chi,\overline{\zeta}+\eta)
\\
&
+\Bigl(\overline{\chi}+\chi+\frac{\rho}{\sqrt{2}}\Bigr)\Bigl(\overline{\chi}+\chi-\frac{\rho}{\sqrt{2}}\Bigr)
\Bigl( 
e^{(\overline{\zeta}-\overline{\chi})(\eta-\chi)}L_{\rho}(\overline{\zeta}+\eta)L_{\rho}(\overline{\chi}+\chi)
-
L_{\rho}(\overline{\zeta}+\chi)L_{\rho}(\overline{\chi}+\eta)
\Bigr)
\\
&
+
\mathcal{C}_{\rho}(\overline{\zeta}+\chi,\overline{\chi}+\eta)
-e^{(\overline{\zeta}-\overline{\chi})(\eta-\chi)}
\mathcal{C}_{\rho}(\overline{\chi}+\chi,\overline{\zeta}+\eta).
\end{align*}
Then, we have 
\begin{align*}
C_K(z,w)\mathfrak{H}_N^{(n,L)}(z,w,\lambda)
\varpi^{(n,L)}(z,\overline{z},w,\overline{w}|\lambda,\overline{\lambda})
=
N\delta_N(1)\frac{H_3(\zeta,\eta,\chi)\varpi(\zeta,\overline{\zeta},\eta,\overline{\eta}|\chi,\overline{\chi})}{(\overline{\zeta}-\overline{\chi})^2(\eta-\chi)^2\mathcal{L}_{\rho}(\overline{\chi}+\chi)}e^{-(\overline{\zeta}-\overline{\chi})(\eta-\chi)}
\bigl( 1 + O(N^{-1}) \bigr).
\end{align*}
On the other hand, similarly, we have 
\[
C_K(z,w)\mathfrak{I}_N^{(n,L)}(\overline{z},w|\overline{\lambda},\lambda)
\varpi^{(n,L)}(z,\overline{z},w,\overline{w}|\lambda,\overline{\lambda})
=
N\delta_N(1)
\frac{L_{\rho}(\overline{\zeta}+\eta)\varpi(\zeta,\overline{\zeta},\eta,\overline{\eta}|\chi,\overline{\chi})}{(\overline{\zeta}-\overline{\chi})(\eta-\chi)}
\bigl( 1 + O(N^{-1}) \bigr).
\]
For \eqref{Part2N} and \eqref{Part3N}, it is straightforward to see that 
\[C_K(z,w)\mathfrak{II}_N^{(n,L)}(\overline{z},w|\overline{\lambda},\lambda)
\varpi^{(n,L)}(z,\overline{z},w,\overline{w}|\lambda,\overline{\lambda})
=
o(N\delta_N(1)), 
\]
and 
\[
C_K(z,w)\mathfrak{III}_N^{(n,L)}(\overline{z},w|\overline{\lambda},\lambda)
\varpi^{(n,L)}(z,\overline{z},w,\overline{w}|\lambda,\overline{\lambda})
=
o(N\delta_N(1)). 
\]
Here, we note that 
\begin{align*}
&\quad 
H_3(\zeta,\eta,\chi)
+(\overline{\zeta}-\overline{\chi})(\eta-\chi)L_{\rho}(\overline{\zeta}+\eta)
e^{(\overline{\zeta}-\overline{\chi})(\eta-\chi)}
\mathcal{L}_{\rho}(\overline{\chi}+\chi)
\\
&=
\mathcal{A}_{\rho}(\overline{\chi}+\chi,\overline{\zeta}+\chi,\overline{\chi}+\eta,\overline{\zeta}+\eta,(\overline{\zeta}-\overline{\chi})(\eta-\chi))
+
\mathcal{B}_{\rho}(\overline{\chi}+\chi,\overline{\zeta}+\chi,\overline{\chi}+\eta)
+
\mathcal{B}_{\rho}(\overline{\chi}+\chi,\overline{\chi}+\eta,\overline{\zeta}+\chi)
\\
&
+
(\overline{\zeta}-\overline{\chi})(\eta-\chi)e^{(\overline{\zeta}-\overline{\chi})(\eta-\chi)}
\mathcal{B}_{\rho}(\overline{\chi}+\chi,\overline{\zeta}+\eta,\overline{\chi}+\chi)
-e^{(\overline{\zeta}-\overline{\chi})(\eta-\chi)}\mathcal{B}_{\rho}(\overline{\chi}+\chi,\overline{\zeta}+\eta,\overline{\chi}+\chi)
\\
&
-e^{(\overline{\zeta}-\overline{\chi})(\eta-\chi)}\mathcal{B}_{\rho}(\overline{\chi}+\chi,\overline{\chi}+\chi,\overline{\zeta}+\eta)
+
\mathcal{C}_{\rho}(\overline{\zeta}+\chi,\overline{\chi}+\eta)
-e^{(\overline{\zeta}-\overline{\chi})(\eta-\chi)}
\mathcal{C}_{\rho}(\overline{\chi}+\chi,\overline{\zeta}+\eta).
\end{align*}
Therefore, we finally obtain 
\begin{align*}
&\quad
\lim_{N\to\infty}\frac{1}{N\delta_N(1)}
K_{1,1}^{(N)}(z,\overline{z},w,\wbar|\lambda,\overline{\lambda})
\\
&=
\frac{\mathcal{H}_{\rho}
(\overline{\chi}+\chi,\overline{\zeta}+\chi,\overline{\chi}+\eta,\overline{\zeta}+\eta,(\overline{\zeta}-\overline{\chi})(\eta-\chi))
}{(\overline{\zeta}-\overline{\chi})^2(\eta-\chi)^2\mathcal{L}_{\rho}(\chi+\overline{\chi})}
\varpi(\zeta,\overline{\zeta},\eta,\overline{\eta}|\chi,\overline{\chi})
e^{-(\overline{\zeta}-\overline{\chi})(\eta-\chi)},
\end{align*}
uniformly for $\zeta,\eta,\chi$ in compact subsets of $\C$. 
This completes the proof of \eqref{D11WEAK} and \eqref{ThmWeak3}. 
Similarly, it suffices to consider the front factor for \eqref{D12WEAK}.
From the proof of the diagonal case and by \eqref{D12ver1}, we have
\begin{align*}
&
\quad - \lim_{N\to\infty}\frac{1}{(N\delta_N(1))^2}
\frac{n(z\overline{w}-\frac{L}{n})}{z\overline{w}(1+z\overline{w})}
\widehat{g}_{N-1}^{(n,L)}(z\overline{w})\widehat{\omega}^{(n,L)}(\overline{w},z)
\mathcal{K}^{(N-1)}(\overline{z},w|z,\overline{w})\widehat{\omega}^{(n,L)}(\overline{z},w)
\\
&
=
-
\frac{\mathcal{H}_{\rho}(\overline{\eta}+\zeta,\overline{\zeta}+\zeta,\overline{\eta}+\eta,\overline{\eta}+\zeta,-(\overline{\zeta}-\overline{\eta})(\zeta-\eta))}{(\overline{\zeta}-\overline{\eta})^2(\zeta-\eta)^2}.
\end{align*}
This completes the proof of \eqref{D12WEAK}.
\end{proof}

\subsection{Proof of Theorem~\ref{MainResult} at the singular origin regime}
\begin{proof}[Proof of Theorem~\ref{MainResult} at the singular origin regime ]
For $\zeta,\eta,\chi\i$ of compact subsets of $\C$, we write 
\[
z=z(\zeta)=\frac{\zeta}{\sqrt{N\delta_N(0)}},\quad 
w=w(\eta)=\frac{\eta}{\sqrt{N\delta_N(0)}},\quad
\lambda=\lambda(\chi)=\frac{\zeta}{\sqrt{N\delta_N(0)}}.
\]
For instance, we have
\[
z\wbar=\frac{\zeta\overline{\eta}}{N(b+1)}\bigl ( 1+O(N^{-1}) \bigr). 
\]
Then, by Lemma~\ref{Lem_Asym}, we have
\begin{equation}
\widehat{g}_N^{(n,L)}(\lambda\overline{\lambda})\widehat{\omega}^{(n,L)}(\lambda,\overline{\lambda})
=
N\frac{(\chi\overline{\chi})^L\mathcal{E}_{1,L}(\chi\overline{\chi})e^{-|\chi|^2}}{\chi\overline{\chi}-L}
\bigl ( 1+O(N^{-1}) \bigr),
\end{equation}
which gives 
\[
\frac{n(\lambda\overline{\lambda}-L/n)}{\lambda\overline{\lambda}(1+\lambda\overline{\lambda})}
\widehat{g}_N^{(n,L)}(\lambda\overline{\lambda})\widehat{\omega}^{(n,L)}(\lambda,\overline{\lambda})
=
N^2(b+1)|\chi|^{2L-2}\mathcal{E}_{1,L}(\chi\overline{\chi})e^{-|\chi|^2}
\bigl ( 1+O(N^{-1}) \bigr).
\]
Now, similar to the other case, we compute the asymptotic behavior for each term. 
For $\mathfrak{H}_N^{(n,L)}(\overline{z},w|\overline{\lambda},\lambda)$, we have
\begin{align}
\begin{split}
\mathfrak{H}_N^{(n,L)}(\overline{z},w|\overline{\lambda},\lambda)
&=
(N(b+1))^2(\overline{\zeta}\eta)^Le^{-\frac{1}{2}(|\zeta|^2+|\eta|^2)}
\\
&\times
\frac{\mathcal{E}_{1,L}(\overline{\zeta}\chi|\overline{\chi}\chi)\mathcal{E}_{1,L}(\overline{\chi}\eta|\overline{\chi}\chi)-\mathcal{E}_{1,L}(\overline{\zeta}\eta|\overline{\chi}\chi)\mathcal{E}_{1,L}(\overline{\chi}\chi|\overline{\chi}\chi)}{(\overline{\zeta}-\overline{\chi})^2(\eta-\chi)^2(\overline{\chi}\chi-L)\mathcal{E}_{1,L}(\chi\overline{\chi}|\overline{\chi}\chi)}
\bigl ( 1+O(N^{-1}) \bigr).
\end{split}
\end{align}
For $\mathfrak{F}_N^{(n,L)}(\overline{z},w|\overline{\lambda},\lambda)$, 
since 
\begin{align*}
\mathfrak{I}_N^{(n,L)}(\overline{z},w|\overline{\lambda},\lambda)
&=
(N(b+1))^2\frac{(\overline{\zeta}\eta)^Le^{-\frac{1}{2}(|\zeta|^2+|\eta|^2)}}{(\overline{\zeta}-\overline{\chi})(\eta-\chi)}
\Bigl(
\frac{\mathcal{E}_{1,L}(\overline{\zeta}\eta|\overline{\chi}\chi)}{\overline{\chi}\chi-L}
-
\frac{1}{\Gamma(L)} 
\frac{1}{\chi\overline{\chi}-L}
\Bigr)
\bigl ( 1+O(N^{-1}) \bigr).
\end{align*}
For \eqref{Part2N} and \eqref{Part3N}, similarly, it is easy to see that 
\begin{align*}
&C_K(z,w)
\mathfrak{II}_N^{(n,L)}(\overline{z},w|\overline{\lambda},\lambda)
\varpi^{(n,L)}(z,\overline{z},w,\overline{w}|\lambda,\overline{\lambda})
=
o(1),
\\
&C_K(z,w)
\mathfrak{III}_N^{(n,L)}(\overline{z},w|\overline{\lambda},\lambda)
\varpi^{(n,L)}(z,\overline{z},w,\overline{w}|\lambda,\overline{\lambda})=
o(1), 
\end{align*}
and 
\[
\varpi^{(n,L)}(\overline{z},w|\lambda,\overline{\lambda})
=
\frac{\varpi(\zeta,\overline{\zeta},\eta,\overline{\eta}|\chi,\overline{\chi})}{(b+1)N}\bigl( 1 + O(N^{-1}) \bigr).
\]
Since
\begin{align*}
&\quad \mathfrak{H}_N^{(n,L)}(\overline{z},w|\overline{\lambda},\lambda)
+
\mathfrak{F}_N^{(n,L)}(\overline{z},w|\overline{\lambda},\lambda)
\\
&=
(N(b+1))^2
\frac{\mathcal{S}_L(\overline{\zeta}\chi,\overline{\chi}\eta,\overline{\zeta}\eta,\overline{\chi}\chi,(\overline{\zeta}-\overline{\chi})(\eta-\chi))}{(\overline{\zeta}-\overline{\chi})^2(\eta-\chi)^2\mathcal{E}_{1,L}(\overline{\chi}\chi|\overline{\chi}\chi)}
(\overline{\zeta}\eta)^{L}e^{-\frac{1}{2}(|\zeta|^2+|\eta|^2)}
(1+O(N^{-1})),
\end{align*}
we obtain 
\begin{align*}
&\quad
\lim_{N\to\infty} \frac{1}{N\delta_N(0)}K_{1,1}^{(N)}(z,\overline{z},w,\wbar|\lambda,\overline{\lambda})
\\
&=
\frac{\mathcal{S}_L(\overline{\zeta}\chi,\overline{\chi}\eta,\overline{\zeta}\eta,\overline{\chi}\chi,(\overline{\zeta}-\overline{\chi})(\eta-\chi))}{(\overline{\zeta}-\overline{\chi})^2(\eta-\chi)^2\mathcal{E}_{1,L}(\overline{\chi}\chi|\overline{\chi}\chi)}
\varpi(\zeta,\overline{\zeta},\eta,\overline{\eta}|\chi,\overline{\chi})
(\overline{\zeta}\eta)^{L}e^{-\frac{1}{2}(|\zeta|^2+|\eta|^2)},
\end{align*}
uniformly for $\zeta,\eta,\chi$ in compact subsets of $\C$. 
This completes the proof of \eqref{D11SINGULAR}.
Similarly, it suffices to consider the front factor for \eqref{D12SINGULAR}.
From the proof of the diagonal case and by \eqref{D12ver1}, we have 
\begin{align*}
&
\quad - \lim_{N\to\infty}\frac{1}{N}\frac{1}{(N\delta_N(0))^2}
\frac{n(z\overline{w}-\frac{L}{n})}{z\overline{w}(1+z\overline{w})}
\widehat{g}_{N-1}^{(n,L)}(z\overline{w})\widehat{\omega}^{(n,L)}(\overline{w},z)
\mathcal{K}^{(N-1)}(\overline{z},w|z,\overline{w})\widehat{\omega}^{(n,L)}(\overline{z},w)
\\
&=
-
\frac{\mathcal{E}_{1,L}(\zeta\overline{\eta})|\zeta|^{2L}|\eta|^{2L}}{\zeta\overline{\eta}}e^{-|\zeta|^2-|\eta|^2}
\cK_{1,1}^{(\mathbf{s})}(\overline{\zeta},\eta|\zeta,\overline{\eta}).
\end{align*}
This completes the proof. 
\end{proof}

\section{Concluding remarks}\label{Section6}
In this note, we studied the $k$-th conditional expectation of the diagonal and off-diagonal overlaps for induced spherical unitary ensemble. 
Although this model is non-Gaussian model with the origin point insertion, 
the universality for the overlap in the strongly non-unitary regime was confirmed. 
Also, in \cite{N23}, the weakly non-unitary and the singular origin regime for induced Ginibre unitary ensemble were investigated. We have confirmed that the scaling limits for their regimes also fall into the universality classes. 
Our approach is based on \cite{ATTZ,N23}, and we would like to emphasize that the method developed in \cite{ATTZ} is robust for non-Hermitian random matrix models with a radially symmetric potential such as truncated unitary ensemble. 
Lastly, we conclude this note by making lists for future directions of the overlaps. 
\begin{enumerate}
\item The most important model in random matrices with non-radially symmetric potential is elliptic Ginibre unitary ensemble defined by 
\[
G_N^{(\tau)}=\sqrt{\frac{1+\tau}{2}}H_N^{(1)}+i\sqrt{\frac{1-\tau}{2}}H_N^{(2)},\quad\text{for $\tau\in[-1,1]$},
\]
where $H_N^{(1)}, H_N^{(2)}$ are identically, independent Gaussian unitary ensembles. 
The joint probability distribution function for the eigenvalues of elliptic Ginibre unitary ensemble is given by 
\[
d\P_N^{(\tau)}(\boldsymbol{\zeta}_{(N)})
=
\frac{1}{Z_N^{(\tau)}}\prod_{1\leq i<j\leq N}|\zeta_i-\zeta_j|^2
\prod_{j=1}^{N}e^{-\frac{N}{1-\tau^2}(|\zeta_j|^2-\tau\re \zeta_j^2)}dA(\zeta_j),
\]
where $Z_N^{(\tau)}$ is the partition function. Similar to the discussion in \cite{AFK20,BD21,D21v1}, it is straightforward to see that conditionally on $\{z_1=a\in\C\}$, we have 
\[
\E_N\bigl[\mathcal{O}_{1,1}|z_1=a \bigr]=\prod_{j=2}^N\Bigl(1+\frac{1-\tau^2}{N}\frac{1}{|z_j-a|^2} \Bigr). 
\]
The conditional expectation of the off-diagonal overlap can be similarly computed. Then, in order to study the similar determinantal structure of the overlaps for elliptic Ginibre unitary ensemble as shown in this note, we need to construct a family of planar orthogonal polynomials associated with the following weight function: 
\[
\omega_\tau(z,\overline{z}|a,\overline{a})=\Bigl(\frac{1-\tau^2}{N}+|z-a|^2 \Bigr)
e^{-\frac{N}{1-\tau^2}(|z|^2-\tau\re z^2)}.
\]
However, the moment matrix associated with the above weight function as in \eqref{moment1} is not a tridiagonal moment matrix, and hence, it would be difficulty to work on LDE decomposition. 
Possibly, we need to choose another nice basis such as Hermite polynomial in advance. 
We emphasize that the diagonal overlap was already considered in \cite{CFW24,CW24,FT21} by the supersymmetric method. 
\item 
In this note, we only studied the spherical unitary ensemble. It is natural to study the Pfaffian structure of the overlaps for induced spherical orthogonal/symplectic ensemble, cf. \cite{AFK20,ABN24}. 
\end{enumerate}

\section*{Acknowledgment}
I am deeply thankful to Gernot Akemann and Sung-Soo Byun for insightful discussions in an early stage. 
This work was started during my stay in Korea Institute For Advanced Study (KIAS) and Bielefeld University. 
I am deeply grateful to these institutions for their warm hospitality. 
This work was supported by WISE program (JSPS) at Kyushu University, JSPS KAKENHI Grant Number (B) 18H01124 and 23H01077, and the Deutsche Forschungsgemeinschaft (DFG) grant SFB 1283/2 2021--317210226.
I am very grateful to two anonymous referees for a number of comments and suggestions which have improved the presentation.

\bibliographystyle{abbrv}
 
\end{document}